\documentclass[11pt]{article}

\usepackage{graphicx}
\usepackage{amsmath, amssymb, color,fullpage,amsthm}
\usepackage{amsfonts}
\usepackage{amssymb}
\usepackage{mathrsfs}
\usepackage{pifont}
\usepackage{enumerate}
\usepackage{color}

\newtheorem{theorem}{Theorem}
\newtheorem{lemma}[theorem]{Lemma}
\newtheorem{claim}[theorem]{Claim}

\newtheorem{observation}{Observation}
\newtheorem{assumption}{Assumption}

\newcommand*{\be}{\begin{equation}}
\newcommand*{\ee}{\end{equation}}
\newcommand*{\bea}{\begin{eqnarray}}
\newcommand*{\eea}{\end{eqnarray}}

\def\bfm#1{\mbox{\boldmath$#1$}}

\DeclareMathOperator{\opt}{opt}

\newcounter{my}

\newcounter{my2}

\newcounter{my3}

\begin{document}

\title{The Price of Anarchy for Selfish Ring Routing is Two}
\author{Xujin Chen\thanks{Institute of Applied Mathematics, AMSS, Chinese Academy of Sciences, Beijing, China.
\texttt{\{xchen,xdhu,mawd\}@amss.ac.cn}}
\and Benjamin Doerr\thanks{Max Planck Institute for Informatics, Saarbr{\"u}cken, Germany. \texttt{\{doerr,vanstee,winzen\}@mpi-inf.mpg.de}} \and Xiaodong Hu\footnotemark[1] \and Weidong Ma\footnotemark[1]  \and Rob van Stee\footnotemark[2]  \and Carola Winzen\footnotemark[2] }
\maketitle

\begin{abstract}
We analyze the network congestion game with atomic players, asymmetric strategies, and the maximum latency among all players as social cost.
This important social cost function is much less understood than the average latency.
We show that the price of anarchy is at most two, when the network is a ring and the link latencies are linear.
Our bound is tight. This is the first sharp
 bound for the maximum latency objective.
\end{abstract}

\section{Introduction}
\label{sec:introduction}

Selfish routing is a fundamental problem in algorithmic game theory, and was one
of the first problems which were intensively studied in this field~\cite{KP99,MS01,RT02,Czu04}.
A main question in this field concerns the cost of selfishness: how much performance is
lost because agents behave selfishly, without regard for the other agents or for any global
objective function?

The established measure for this performance loss is
the price of anarchy (PoA)~\cite{KP99}. This is the worst-case ratio
between the value of a Nash equilibrium, where no player can deviate unilaterally to
improve, and the value of the optimal routing.

 Of particular interest to computer science are network congestion games, where agents choose routing paths and experience delays (latencies) depending on how much other players also use the edges on their paths.
Such games are guaranteed to admit at least one Nash equilibrium \cite{R73}.
Generally, the price of anarchy for a selfish routing problem
may depend on the network topology, the number of
players (including the \emph{non-atomic} case where an infinite number of players each controls a negligible fraction of the solution), the type of latency functions on the links, and the objective functions
of the players and of the system (the latter is often called the \emph{social cost function}).

Most of the existing research has focused on the price of anarchy for
minimizing the \emph{total} latency of all the players~\cite{rough02,monien06}.
Indeed, this measure is so standard that it is often not even mentioned
in titles or abstracts. In most cases, a symmetric setting was considered
where all players have the same source node and the same destination node,
and hence the same strategy set. \cite{ck05} and   \cite{aae05}  independently proved that the PoA of the atomic congestion game (symmetric or asymmetric) with linear latency is at most 2.5. This bound is tight. The bound grows to 2.618 for weighted demands \cite{aae05}, which is again a tight bound.
In non-atomic congestion games with linear latencies,
the PoA is at most 4/3~\cite{RT02}. This is witnessed already by two parallel links.
The same paper also extended this result to polynomial latencies.

In this work, we regard as social cost function the \emph{maximum} latency a player experiences. While this cost function was suggested already in~\cite{KP99}, it seems much less understood. For general topologies, the maximum PoA of atomic congestion games with linear latency is 2.5 in single-commodity networks (symmetric case, all player choose paths between the same pair of nodes),
but it grows to $\Theta(\sqrt k)$ in $k$-commodity networks 
(asymmetric case, $k$ players have different nodes to connect via a path) \cite{ck05}. The PoA further increases with additional restrictions to the strategy sets. \cite{glmm06} showed that when the graph consists of $n$ parallel links and each player's choice can be restricted to a particular subset of these links,  the maximum PoA lies in the interval $[n-1,n)$.

For non-atomic  selfish routing, \cite{lrtw11} showed that the PoA of symmetric games on  $n$-node networks with arbitrary continuous and non-decreasing latency functions is $n-1$, and exhibited an infinite family of asymmetric games whose PoA grows exponentially with the network
size.

\paragraph{Our setting:}
In this work, we analyze the price of anarchy of a maximum latency network congestion game for a concrete and useful network topology, namely rings.
Rings are frequently encountered in communication networks.
Seven self-healing rings form the  EuroRings network, the largest, fastest, best-connected high-speed network in Europe, spanning 25,000 km and connecting 60 cities in 18 countries. As its name suggests, the Global Ring Network for Advanced Applications Development (GLORIAD) 
\cite{gloriad} is an advanced science internet network constructed as an optical ring around the Northern Hemisphere. The global ring topology of the network provides scientists, educators and students with advanced networking tools, and enables active, daily collaboration on common problems. 
It is therefore worthwhile to study this
topology in particular. Indeed, considerable research has already
gone into studying rings, in particular in the context of designing
approximation algorithms for combinatorial optimization
problems~\cite{AnsZha08,BlKaKl01,Cheng04,ScSeWi98,Wang05}.

As in most previous work, we assume that traffic may not be split, because
this causes the problem of packet reassembly at the receiver and is therefore
generally avoided. Furthermore,
we assume that the edges (``links'') have linear latency functions.
That is, each link $e$ has a latency function
$\ell_e(x)=a_e x+b_e$, where $x$ is the number of players using link
$e$ and $a_e$ and $b_e$ are nonnegative constants.

For the problem of minimizing the maximum latency, even assuming a central authority,
the question of how to route communication requests optimally is nontrivial;
it is not known whether this problem is in $P$.
It is known for general (directed or undirected) network topologies that already the price of stability (PoS), which is the ratio of the value of
the \emph{best} Nash equilibrium to that of the optimal solution~\cite{adktwr},
is unbounded for this goal function
even for linear latency functions~\cite{ChChHu10,ChChHH11}.
However, this is not
the case for 
rings. It has been shown that for any instance on a ring, either its PoS equals 1, or its PoA is at most 6.83, giving a universal upper bound 6.83 on PoS for the selfish ring routing \cite{ChChHu10}. The same paper also gave a lower bound  of 2  on the PoA. Recently, an upper bound of $16$ on the PoA was obtained \cite{ChChHH11}.

\paragraph{Our results:}
In this paper, we show that the PoA for minimizing the maximum latency 
on rings is exactly 2. This improves  upon the previous best known upper bounds
on both the PoA and the PoS \cite{ChChHH11,ChChHu10}.
Achieving the tight bound required us to upper bound a high-dimensional nonlinear optimization problem.
Our result implies that the performance loss due to selfishness is relatively low for this problem.
Thus, for ring routing, simply allowing each agent to choose its
own path will always result in
reasonable performance. The lower bound example (see Figure \ref{fg:instance})
can be modified to give a lower bound of $2^d$ for
latency functions that are polynomials of degree at most $d$. 

\paragraph{Proof overview:}
Our proof consists of two main parts: first, we analyze
for Nash equilibria the maximum ratio
of the latency of any player to the latency of the entire ring, and
then we analyze the ratio of
the latency of the entire ring in a Nash equilibrium to the
maximum player latency in an optimal routing.
In the first part we show that this ratio
is at most roughly
$2/3$; the precise value depends on whether or not
every link of the ring is used by at least one player in the
Nash equilibrium.

For the second ratio, we begin by showing the very helpful fact
that it is sufficient to consider only instances where no player uses the same path in the Nash
routing as in the optimal routing.
For such instances,
we need to distinguish two cases. The first case deals with instances for which there exists a link 
that in the Nash equilibrium is not used by
any player. For such instances we use a structural analysis
to bound the second ratio from above by $2+2/k$, where $k$ is the number of agents in the system.

For the main case in which the paths of the players in the Nash equilibrium
cover the ring, we show that the second ratio is at most 3. We begin by using the standard technique of adding up the Nash inequalities
which state that no player can improve by deviating to its
alternative path. This gives us a constraint which must be
satisfied for any Nash equilibrium, but this does not immediately
give us an upper bound for the second ratio.
Instead, we end up with a nonlinear optimization problem: maximize the ratio
under consideration subject to the Nash constraint.
The analysis of this problem was the main technical challenge of this paper.
We use a series of modifications to reach an optimization problem with only
five variables, which, however, is still nonlinear.
It can be solved by Maple, but we also provide a formal solution.

\section{The Selfish Ring Routing Model}
\label{sec:models}

Let $\mathcal I=(R, \ell,(s_i,t_i)_{i \in [k]})$ be a selfish ring routing (SRR) instance, where $R=(V,E)$ is a ring and where for each agent $i \in [k]$ the pair $(s_i,t_i)$ denotes the source and the destination nodes of agent $i$.
We sometimes refer to the agents as \emph{players}.
For every \emph{link} $e\in E$ we denote the \emph{latency function} by
$\ell_e(x)=a_e x+b_e$, where $a_e$ and $b_e$ are nonnegative constants; without loss of generality we assume that $a_e$, $b_e$ are nonnegative integers. This is feasible since real-valued inputs can be approximated arbitrarily well by integers by scaling the input appropriately.

For any subgraph $P$ of $R$ (written as $P\subseteq R$), we slightly abuse the notation and identify $P$ with its link set $E(P)$.
If $Q$ is a path on $R$ {with end nodes $s$ and $t$}, we use $P\backslash Q$ to denote the graph  obtained from $P$ by removing {all nodes in $V(P)\cap V(Q)\setminus\{s,t\}$} (all internal nodes of $Q$ which are contained in $P$), and all links in $P \cap  Q $ (all links of $Q$ which are contained in $P$).

For any feasible routing $\pi=\{P_1, \ldots, P_k\}$, where $P_i$ is a path on $R$ between $s_i$ and $t_i$, $i=1,\ldots,k$, we
denote by $M(\pi):=\max_{i \in [k]}{\ell(P_i,\pi)}$ the maximum latency of any of the $k$ agents.
Here we abbreviate by
$\ell(P, \pi)$ the latency $$\ell(P, \pi):= \sum_{e \in P}\left({a_e |\{i \in [k] \mid e \in P_i \}| + b_e}\right)$$ of a subgraph $P \subseteq R$ in $\pi$.
We say that $\pi$ is a {\em Nash equilibrium} ({\em routing}) if no agent $i\in [k]$ can reduce its latency $\ell(P_i,\pi)$ by switching $P_i$ to the alternative path $R\backslash P_i$, provided other agents do not change their paths.

Sometimes we are only interested in the latency caused by one additional agent and we write 
$||P||_a:=\sum_{e \in P}{a_e}$.
Similarly we abbreviate $||P||_b:=\sum_{e \in P}{b_e}$.

Let $\pi^N=\{N_1,\ldots,N_k\}$ be some fixed worst Nash routing (i.e., a Nash equilibrium with maximum system latency $M(\pi^N)$), and let $\Pi^*$ be the set of optimal routings of $\mathcal{I}$.

For any $\pi =\{Q_1,\ldots, Q_k\}\in\Pi^*$,
let $$h(\pi):=|\{i\in[k]:N_i\not= Q_i\}|.$$ I.e., $h(\pi)$ is the number of agents for which
their Nash routings are not the same as their optimal routings.
We choose $\pi^*=\{Q_1,\ldots,Q_k\}\in\Pi^*$
to be an optimal routing that minimizes {$h=h(\pi^*)$}.
Without loss of generality, we assume that $\{i\in[k]:N_i\not= Q_i\}=[h]:=\{1, \ldots, h\}$. We call the agents $1, \ldots, h$ \emph{switching} agents and we refer to the agents in $[k]\backslash [h]$ as \emph{non-switching} ones.

For brevity, we write $\ell^*(P):=\ell(P, \pi^*)$ and $\ell^N(P):=\ell(P, \pi^N)$.
Abusing notation, for any link $e \in R$, we set
$$\pi^*(e):=|\{i\in [h] \mid e \in Q_i\}|,$$
the number of \emph{switching} (!) players whose optimal paths traverse $e$.
Analogously, $\pi^N(e):=|\{i\in [h] \mid e \in N_i\}|$.

\section{Main Result and Outline of the Proof}
\label{sec:outline}

The purpose of this paper is the proof of the following statement.
\begin{theorem}
\label{poa}
The price of anarchy for selfish ring routing with linear latencies is~$2$.
\end{theorem}

As mentioned in the introduction, a simple example for which the price of anarchy is two has been given already in~\cite{ChChHu10}. This is the example given in Figure~\ref{fg:instance}. As is easy to verify, $M(\pi^*)=1$ and $M(\pi^N)=~2$.
\begin{figure}[htpb]
\centerline{\includegraphics[scale=0.64]{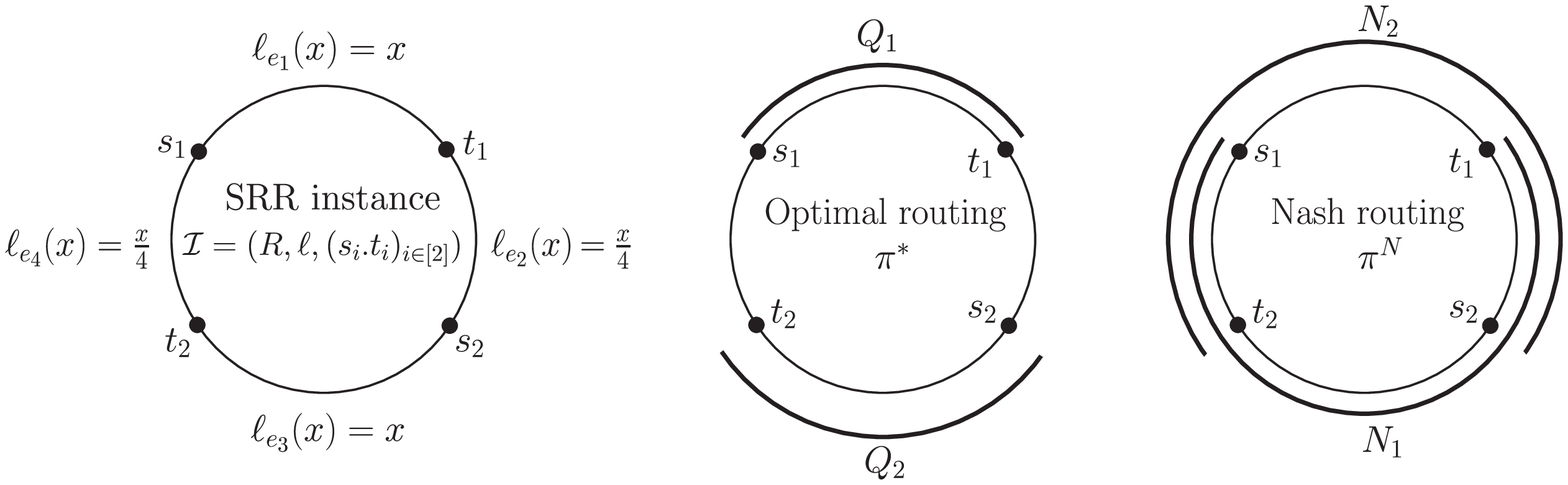}}
\caption{ \label{fg:instance} A 2-player  SRR instance with PoA = 2. }
\end{figure}

Hence, our result is tight. We can resort to proving the upper bound in Theorem~\ref{poa}. That is, we need to show that for all SRR instances $\mathcal{I}$ the ratio $M(\pi^N)/M(\pi^*)$ is at most two. The main steps are as follows.
\begin{enumerate}
	\item We begin by restricting the set of Nash routings we need to consider.
We show that we can assume without loss of generality that in $\pi^N$
there is at most \emph{one} player that uses the same path as in  $\pi^*$, i.e., $h\ge k-1$ (Section~\ref{sec:relationkh}).
We call the case where there is such a player the \emph{singular} case; if there is no such a player,
we are in the \emph{nonsingular} case.
	\item We say that the Nash equilibrium $\pi^N$ is a \emph{covering} equilibrium if the Nash paths of the switching agents $1, \ldots, h$ cover the ring, i.e., if $\cup_{i \in [h]}N_i = R$.
For any {non-covering equilibrium}, we use a structural analysis of $\pi^N$ to
show (Section \ref{sec:uncovered})
that the PoA is less than two for $h \geq 3$. 
	\item We proceed by showing (Lemma~\ref{lem:23}) that for every {covering equilibrium},
the ratio $M(\pi^N)/\ell^N(R)$ is at most  $2/3$.
	\item Finally, in the remainder of Section~\ref{sec:covering}, we show that $\ell^N(R)/M(\pi^*,I)\le3$ for any covering equilibrium $\pi^N$.
This is the main part of the proof. 
Combining this with the third statement concludes the proof of Theorem~\ref{poa} for covering equilibria.
\end{enumerate}

Some specific cases with small values of $h$ need to be handled separately. 
Our proof needs the following technical lemma which is true for both covering and  non-covering equilibria.
It shows that any two Nash paths of {agents} that use different paths
in $\pi^N$ and in $\pi^*$ share at least one common link. 

\begin{lemma}
\label{lem:twopaths}
For all $i,j \in [h]$, $N_i$ and $N_j$ are not  link-disjoint.
\end{lemma}
\begin{proof}
Assume there exist two agents $i, j \in [h]$ such that $N_i$ and $ N_j$ have no link
 in common. Hence their complements, the optimal paths $Q_i$ and $Q_j$
jointly cover the entire ring, that is, $ Q_i \cup  Q_j =R$.

Consider the routing $\pi'$ which is exactly the same as $\pi^*$, except for
these two agents who use their Nash paths $N_i, N_j$ instead.
For any link $e\in  Q_i\cap  Q_j $ we have $\pi'(e)=\pi^*(e)-2$, and for every link $e\in ( Q_i\backslash  Q_j)\cup ( Q_j\backslash  Q_i)$ the number of agents on this link does not change, i.e.,  $\pi'(e)=\pi^*(e)$.
Since $a_e\geq 0$ for all $e \in E$, this yields $M(\pi')\le M(\pi^*)$.
Hence, $\pi'\in \Pi^*$.
But we also have $h(\pi') < h(\pi^*)$, contradicting the choice of $\pi^*$ given in Section \ref{sec:models}.
\end{proof}

\subsection{Reduction to Singular and Nonsingular Instances}
\label{sec:relationkh}

\begin{lemma}
\label{lem:singular}
Consider any selfish ring routing instance $\mathcal I=(R, \ell,(s_i,t_i)_{i \in [k]})$ with linear latencies. Let $\pi^*$ be an optimal routing and let $\pi^N$ be a Nash routing.
Suppose there is an agent $q \in [k]$ that
uses the same path in $\pi^N$ as in $\pi^*$.
Then there exists a selfish routing instance $\mathcal I'=(R, \ell',(s_i,t_i)_{i \in [k]\backslash\{q\}})$
with linear latency functions $\ell_e'(x)$
such that
\begin{itemize}
	\item the non-switching agent $q$ is removed from $\mathcal{I}$ to get $\mathcal I'$,
	\item the routing $\pi^N$ restricted to the remaining agents, denoted as ${\pi^N}'$, 
is a Nash equilibrium for $\mathcal{I'}$,
	\item the total ring latencies satisfy ${\ell'}^N(R):= \ell'(R,{\pi^N}')=\ell^N(R)$, 
and
	\item we have $M'(\opt')\leq M(\pi^*)$ for the maximum latencies of individual agents. Here, $\opt'$
denotes an optimal routing for $\mathcal{I'}$ and
$M'(\cdot)$ denotes the maximum latency of a routing in $\mathcal{I'}$.
\end{itemize}
\end{lemma}

\begin{proof}
By definition, player $q$ uses path $Q_q$ in both $\pi^N=\{N_i:i\in[k]\}$ and $\pi^* =\{Q_i:i\in[k]\}$.
Remove player $q$ from $\mathcal{I}$. For every link $e\in Q_q$ set
$\ell_e'(x):=\ell_e(x)+a_e=a_e x  + b_e + a_e$. The latency functions of all other links are  unchanged. 
Denote the resulting instance $(R, \ell',(s_i,t_i)_{i \in [k]\backslash\{q\}})$ by $\mathcal{I}'$.

Every routing $\pi$ for $\mathcal{I}$ induces a routing $\pi'$ for $\mathcal{I}'$  in the natural way, by omitting the routing for player $q$. From the modified latency defined in the proof, we see that the latency of every edge in an induced routing is the same as the original latency in $\mathcal{I}$. It follows immediately that
\begin{itemize}
\item
a routing which is a Nash equilibrium in $\mathcal{I}$ induces a Nash equilibrium routing in $\mathcal{I}'$,
\item
the latency of the entire ring of an induced routing is also the same as the ring latency of the original routing in $\mathcal{I}$, and
\item the maximum latency of the induced routing ${\pi^*}'$ of the optimal routing $\pi^*$ is not larger than
the maximum latency of the optimal routing itself, i.e., $M'({\pi^*}')\leq M(\pi^*)$.
\end{itemize}

By definition, the \emph{optimal} routing $\opt'$ for instance $\mathcal{I}'$ cannot be worse than the feasible routing ${\pi^*}'$, and we conclude $M'(\opt') \le M'({\pi^*}') \le  M(\pi^*)$.
\end{proof}

We call the Nash routing $\pi^N$ {\em singular} if $M(\pi^N) > \max_{i \in [h]}\ell^N(N_i)$, i.e., if the maximum latency in $\pi^N$ is obtained only by an agent which uses the same routing in $\pi^N$ as it uses in $\pi^*$.
We call $\pi^N$ {\em nonsingular} otherwise.
That is, $\pi^N$ is nonsingular if $M(\pi^N)=\max_{i \in [h]}\ell^N(N_i)$.
Since we are interested in upper bounding the ratio $M(\pi^N)/M(\pi^*)$,
{applying Lemma \ref{lem:singular} repeatedly enables us to make the following assumption.}

\begin{assumption}
\label{ass1}
$h\le k\le h+1$ and $h=k+1$ if and only if $\pi^N$ is singular.
\end{assumption}

{Under Assumption \ref{ass1}, for any singular case $(\pi^N,\mathcal I)$, Lemma~\ref{lem:singular} produces a nonsingular case $(\pi^{N'},\mathcal I')$ with $\ell'^N(R,\mathcal{I'})/M'(opt',\mathcal{I'})\ge \ell^N(R,\mathcal{I})/M(\pi^*,\mathcal{I})$. Therefore we can upper bound the price of anarchy for {the SRR} problem as follows:}
\begin{itemize}
 \item analyze the ratio $\ell^N(R,\mathcal{I})/M(\pi^*,\mathcal{I})$
only for nonsingular instances $\mathcal{I}$ where no player uses the same path in $\pi^N$ and $\pi^*$,
and
\item analyze the ratio $M(\pi^N,\mathcal{I})/\ell^N(R,\mathcal{I})$ for general instances $\mathcal I$.
\end{itemize}

This is what we will do in the remainder of the paper.

\section{Non-Covering Equilibria}
\label{sec:uncovered}

\begin{theorem}
\label{uncover}
The ratio
$M(\pi^N)/M(\pi^*)$ is at most $\frac{4}{3}+\frac{5}{3h}$ for instances
for which $\cup_{i \in [h]} N_i \neq R$.
\end{theorem}

The proof of Theorem~\ref{uncover} consists of the following two steps.
First we show that the ratio $\ell^N(R)/M(\pi^*)$ is at most $2+\frac{2}{h}$. This is Lemma~\ref{lem:ringnotcovered}.
Next we show (Lemma~\ref{lem:uncoveredalpha}) that for any uncovered instance, if $\ell^N(R)/M(\pi^*) \leq \alpha$ for some constant $\alpha$, then $M(\pi^N)/M(\pi^*)$ is at most $(2\alpha+\frac{1}{h})/3$.
This proves Theorem~\ref{uncover}, which itself proves Theorem~\ref{poa} for the non-covered case with $h \geq 3$. The remaining case of non-covering equilibria with $h=2$ is handled in Section \ref{sec:h=12},
where we show $M(\pi^N)/M(\pi^*)\le2$ directly by utilizing  the structural properties of rings.

\begin{lemma}
\label{lem:ringnotcovered}
Let $\mathcal{I}$ be an SRR instance with $\cup_{i \in [h]}N_i\ne R$.
Then $\ell^N(R)/M(\pi^*)\le 2+\frac{2}{h}$.
\end{lemma}

\begin{proof}
By Lemma~\ref{lem:singular} it suffices to consider the nonsingular case. That is, we assume without loss of generality that $k=h$, i.e., we assume that all agents change their paths.
There exist two agents $p,q\in [h]$ such that
$\cup_{i \in [h]}N_i \subseteq N_p \cup  N_q \varsubsetneq R$, and all $h$ paths in $N_1,N_2,\ldots,N_h $ share a common link in $N_p\cap N_q$.
This holds because if there were three agents that do not all share a same link, then two of them
would not share a link at all. This is due to the assumption $\cup_{i \in [h]}N_i \neq R$.
However, this contradicts Lemma~\ref{lem:twopaths}.
Therefore we can take $P$ to be the longest path in $N_p\cup N_q$ with end link $g_1$ and $g_2$ (possibly $\{g_1\}=\{g_2\}= P$) such that $\pi^N(g_i)>h/2$ for $i=1,2$ and
 \be
 \pi^N(g)\le h/2\text{ for any link }g\in  N_p\cup N_q\setminus P.\label{any1}
 \ee
\begin{figure}
\centerline{
\scalebox{1}{
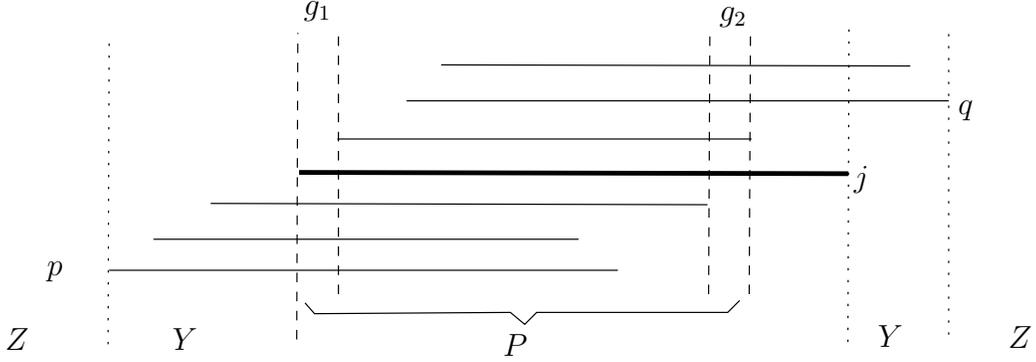}
}
\caption{\label{fig:noncovering}Proof for non-covering equilibria.
For this figure, we have mapped the ring to the real line.}
\end{figure}
See Figure \ref{fig:noncovering}.
Since we have $g_1=g_2$ or $\pi^N(g_1)+\pi^N(g_2)>h$, there exists an agent $j\in [h]$ such that $\{g_1,g_2\}\subseteq  N_j$ and thus $P\subseteq N_j$. Let $Y\subseteq  Q_j$ consist of links $e$ with $\pi^N(e)\geq1$ and $Z= Q_j \backslash Y$.
It can be seen from   {(\ref{any1}) that
$\ell^N(Q_j)\le \frac{h}2 ||Y||_a+||Y||_b+||Z||_b$
and therefore
\begin{align}
\nonumber
\ell^N(R)&=\ell^N(Q_j)+\ell^N(N_j)\le 2\ell^N(Q_j)+||Y||_a+||Z||_a\\
\label{eq:ublrfn}
& \le  (h+1)||Y||_a + 2||Y||_b + ||Z||_a + 2||Z||_b.
\end{align}}
Since
\begin{equation}
\label{eq:lb-lqjf*}
\ell^*(Q_j)\geq \frac{h}2 ||Y||_a +||Y||_b + h||Z||_a + ||Z||_b,
\end{equation}
the ratio of the upper bound (\ref{eq:ublrfn}) for $\ell^N(R)$ 
to the lower bound (\ref{eq:lb-lqjf*}) for $\ell^*(Q_j)$ is maximized for
$||Z||_a=||Z||_b=||Y||_b=0$ and is 
 $(h+1)/(h/2)=2+2/h$.~
\end{proof}

To conclude the proof of Theorem~\ref{uncover}, we finally show the following.
\begin{lemma}
\label{lem:uncoveredalpha}
The ratio $M(\pi^N)/M(\pi^*)$ is at most $(2\alpha+\frac{1}{h})/3$
for instances for which $\cup_{i \in [h]}N_i\ne R$ and $\ell^N(R)/M(\pi^*)\le \alpha$.
\end{lemma}

\begin{proof}
It suffices to show that for any agent $i \in [k]$ the inequality
$\ell^N(N_i)\leq \frac13(2\alpha+\frac{1}{h}) M(\pi^*)$ holds.
Consider an arbitrary agent $i\in [k]$.
Let $C_i:=R\backslash N_i$, the complement of player $i$'s path $N_i$.
We partition the link set of  $ C_i $ into the set of links
$Y:=\{ e \in  C_i  \mid \pi^N(e) \geq 1\}$ which, in routing $\pi^N$, have at least one agent on it and the set of links $Z:= C_i\backslash Y$ with {no players} on it in routing $\pi^N$.

Since $h$ is the number of players whose paths in $\pi^N$ deviate from the one in $\pi^*$, the   {links} $e$ in $Z$ satisfy $\pi^*(e) \geq h$, that is, there are at least $h$ {players} using these links in the routing $\pi^*$.
Hence $M(\pi^*)\geq h||Z||_a$.
In the routing $\pi^N$,
if player $i$ would switch from path $N_i$ to $C_i$, it would have a latency of at most
$\ell^N(C_i)+||Y||_a+||Z||_a$.
Since $\pi^N$ is a Nash equilibrium, we have
\begin{equation}
\label{eq:nibound}
\ell^N(N_i)\le \ell^N(C_i)+||Y||_a+||Z||_a
\le 2\ell^N(C_i)+\frac{1}{h}M(\pi^*).
\end{equation}
By assumption we also have $\ell^N(N_i)+\ell^N(C_i)=\ell^N(R)\leq\alpha M(\pi^*)$.
Adding twice this inequality to (\ref{eq:nibound}) gives
 $3\ell^N(N_i)\le (2\alpha+\frac{1}{h}) M(\pi^*)$, as required.
\end{proof}

\section{Covering Equilibria}
\label{sec:covering}
For covering equilibria, we show that the price of anarchy is at most $2$. This is again a two-step approach. First, the covering property implies an upper bound $2/3$ on $M(\pi^N)/\ell^N(R)$ as follows. 

\begin{lemma}\label{lem:23}
If $\cup_{i \in [h]}N_i= R$, then
$M(\pi^N)/\ell^N(R)\le 2/3$.
\end{lemma}
\begin{proof}
Take $Q\in \pi^N$ with $\ell^N(Q)=M(\pi^N)$. Then $\ell^N(Q)\le \ell^N(R\backslash Q)+||R\backslash Q||_a$ as $\pi^N$ is covering. From  $  \ell^N(R)=\ell^N(Q)+\ell^N(R\backslash Q)
\ge 2\ell^N(Q)-||R\backslash Q||_a\ge2\ell^N(Q)-
\ell^N(R\backslash Q)=3\ell^N(Q)- \ell^N(R)$, we deduce that $M(\pi^N)=\ell^N(Q)\le\frac{2}{3}\ell^N(R)$.
\end{proof}

Second, we prove $\ell^N(R)/M(\pi^*)\le 3$ by distinguishing between the case $h\le2$ and $h>2$.
\begin{theorem}\label{bound}
If $\cup_{i \in [h]}N_i= R$, then $\ell^N(R)/M(\pi^*)\le 3$.
\end{theorem}

The former case $h\le2$ is proved in Section \ref{sec:h=12}, 
which along with Lemma \ref{lem:firstratio3} in this section establishes   Theorem~\ref{bound}.

By Lemma \ref{lem:singular}, we only need to bound ratio $\ell^N(R)/M(\pi^*)$ for nonsingular case where $h=k$.
In this section we consider the $k=h\ge3$ switching players. For each switching player $i\in[h]
$, we can formulate an inequality
$\ell^N(N_i) \leq \ell^N(Q_i)+ ||Q_i||_a$
saying
that its Nash path may not have a longer latency than its alternative path, if one unit load  is
added on every   {link} of the latter. We obtain a constraint by adding up all of these
inequalities.

We can assume that every {link} has a latency function of $x$ or $1$.
This can be achieved by replacing   {a link $e$}
with latency function $a_ex+b_e$ by $a_e$   {links} with latency function $x$
followed by $b_e$ links with latency function $1$.
Now there are only two types of links left, the ones with latency function $x$ and the ones with latency $1$. We introduce variables which count the number of links of both types which are used by a certain number of players, and write the constraint
that we constructed above in terms of these variables.
We then give an upper bound for $\ell^N(R)/M(\pi^*)$ in terms of these variables as well.

We end up with a nonlinear optimization problem: maximize the ratio
under consideration subject to the Nash constraint.
For this problem, we first show that, for the links with latency function  1, only the \emph{total} number
of players on all these links affects the upper bound. 
For any fixed number of players $h$ that do not use the same path in the Nash routing as in the optimal routing, this still leaves us with $h+3$ variables,
since we have one variable for each possible number of players 
on the links with latency function $x$. We now use a centering
argument to show that only at most two of these $h$ variables are nonzero
in an optimal solution of this optimization problem.

Using normalization, this finally gives us an optimization problem with five variables. This problem unfortunately is still not linear. It can be solved by Maple, but we also provide a
formal solution. To do this, we fix $h$ and another variable, and
solve the remaining problem; we then determine the optimal overall values
of the fixed $h$ and that variable.

\paragraph{Summing the Nash inequalities}
For a given path $P\subseteq R$, let $P^a$ be the subset of   {links}
with   {latency function $x$  and let $P^b$ be the subset of links with latency function 1}.

Consider   {a link} $e \in R^a$ ({resp.} $R^b$).
By definition and our assumption that $k=h$, this link occurs in   {$\pi^N(e)$} Nash paths.
That is, this   {link} occurs $\pi^N(e)$ times on the left-hand side of the $h$
Nash inequalities given above---each time with coefficient $\pi^N(e)$ ({resp. 1}). On the other hand,
it occurs $h-\pi^N(e)$ times on the right-hand side of the inequalities, each time
with coefficient $\pi^N(e)+1$ ({resp. 1}).

Formally, we have for $i=1,\dots,h$
\begin{align*}
 \sum_{e \in N_i^a}\pi^N(e)+\sum_{e\in N_i^b} 1 = \ell^N(N_i) &\leq \ell^N(Q_i) + ||Q_i||_a
= \sum_{e \in Q_i^a}(\pi^N(e)+1)+\sum_{e\in Q_i^b} 1
\end{align*}
and, by summation,
\begin{align*}
\sum_{e \in R^a}(\pi^N(e))^2+\sum_{e\in R^b}\pi^N(e)
\leq \sum_{e \in R^a}(h-\pi^N(e))(\pi^N(e)+1)
+\sum_{e\in R^b} (h-\pi^N(e))\,,
\end{align*}
or 
$\displaystyle\sum_{e\in R^a}{\left(2(\pi^N(e))^2-h\right)} + \sum_{e\in R^b}2\pi^N(e)
\leq \sum_{e\in R^a}(h-1)\pi^N(e)+\sum_{e\in R^b} h.$

Writing $A_i$ (resp. $B_i$) as the number of   links with $i$ players on it and a latency function of $x$ (resp.    $1$),
we can group   {links}
with the same numbers of players and write the above as
\begin{align}
\label{eq:eibound}
\sum_{i=1}^{h} ((2i^2-h)A_i +2iB_i)
&\leq \sum_{i=1}^{h} ((h-1)iA_i+hB_i) 
\\
\label{eq:eibound3}
\Rightarrow\sum_{i=1}^{h} \left(\left(\frac{2i}{h}-\frac{1}{i}\right) C_i
+\frac{2i}{h^2}B_i\right)
& \leq \sum_{i=1}^{h} \left(\frac{h-1}{h}C_i+\frac{1}{h}B_i\right)
\end{align}
where we have written $C_i = \frac{i}hA_i$ and divided by $h^2$.

\paragraph{Bounding the optimal latency}
For the optimal routing we also have, by definition and the fact that we are in the nonsingular case,
$h$ inequalities of the form $M(\pi^*)\geq \ell^*(Q_i)$, $i\in[h]$.
 Summing all the inequalities and dividing
by $h$ implies a lower bound on $M(\pi^*)$, namely
\begin{align*}
M(\pi^*)
 \geq \frac{1}{h}\sum_{i=1}^h{\ell^*(Q_i)}
=\frac{1}{h} \sum_{i=1}^{h}  \left((h-i)^2 A_i +(h-i)B_i\right).\end{align*}
Thus we have
\begin{equation}
 \label{eq:lnrmf}
\frac{\ell^N(R)}{M(\pi^*)}\leq
\frac{
\sum_{i=1}^{h}  \left(i A_i+B_i\right)}
{\sum_{i=1}^{h}  \left(\frac{(h-i)^2}{h} A_i +\frac{h-i}{h}B_i\right)}
=
\frac{ \sum_{i=1}^{h} \left(C_i+\frac1hB_i\right)}
     { \sum_{i=1}^{h} \left(\frac{(h-i)^2}{ih} C_i +\frac{h-i}{h^2}B_i\right)}
\end{equation}
and we want to find an upper bound for this expression under the restriction~(\ref{eq:eibound3}).

\begin{lemma}
If $\sum_{i=1}^{h} C_i = 0$, then $\ell^N(R)/M(\pi^*)\le2$.
\end{lemma}
\begin{proof}
Since $C_i \geq 0$ by definition, we have $C_i=0$ for all $i\in[h]$.
Condition (\ref{eq:eibound3}) implies that $\sum_{i=1}^{h} \frac{i}{h}B_i \le
\frac{1}{2}\sum_{i=1}^h B_i$.
Therefore, by (\ref{eq:lnrmf}), the ratio $\ell^N(R)/M(\pi^*)$ is
at most
$(\sum_{i=1}^h B_i)/(\sum_{i=1}^h{B_i} - \sum_{i=1}^h{\frac{i}{h}B_i})\leq(\sum_{i=1}^h B_i)/(\frac12 \sum_{i=1}^h B_i)=~2$.
\end{proof}

\paragraph{Rewriting the problem} Henceforth we assume $\sum_{i=1}^{h} C_i > 0$. Using $\frac{(h-i)^2}{ih}=\frac{h}{i}+\frac{i}{h}-2$, from (\ref{eq:lnrmf}) we
arrive at the following inequality after dividing numerator and denominator by
$\sum_{j=1}^{h} C_j>0$.
\begin{align*}
\frac{\ell^N(R)}{M(\pi^*)}  & \leq
\frac{1 + \sum_{i=1}^{h} \frac{B_i}{h\sum_{j=1}^h C_j}}
{ \sum_{i=1}^{h} \left(\left(\frac{h}{i}+\frac{i}{h}\right) \frac{C_i}{\sum_{j=1}^h C_j} +\frac{h-i}{h^2}\frac{B_i}{\sum_{j=1}^h C_j}\right)-2}\\
& \leq \frac{ 1+\beta}
     { \sum_{i=1}^{h} \left(\frac{h}{i}+\frac{i}h\right) D _i -2+\beta-{z} }
\end{align*}
where $\beta:=\frac{\sum_{i=1}^{h} B_i}{h\sum_{j=1}^h C_j}\geq0$,
{${z} :=\sum_{i=1}^{h} \frac{iB_i}{h^2\sum_{j=1}^h C_j}\in[\frac\beta{h},\beta]$},
and $D_i:=\frac{C_i}{\sum_{j=1}^h{C_j}}$ for every $i\in[h]$.
Notice that $\sum_{i=1}^h{D_i}=1$.
We divide both sides of (\ref{eq:eibound3}) by $\sum_{j=1}^{h} C_j $ and obtain the constraint $\sum_{i=1}^{h} \left(\frac{2i}h-\frac{1}i\right) D _i +2{z} \leq \frac{h-1}{h}+\beta$.
Our problem now looks as follows.
\begin{align}
\label{eq:optim0}
 \frac{\ell^N(R)}{M(\pi^*)}\le\max &
\frac{ 1+\beta}
     { \sum_{i=1}^{h} \left(\frac{h}{i}+\frac{i}h\right) D_i -2+\beta-{z}}\\
\label{eq:optim1a}
\text{s.t.} & \sum_{i=1}^{h} \left(\frac{2i}h-\frac{1}i\right) D_i +2{z}
	 \leq \frac{h-1}{h}+\beta\\
\label{eq:optim2a}
&  \sum_{i=1}^{h} D_i=1,\qquad D_i\geq0\ \forall i\in[h]\\
\label{eq:optim3a}
&\beta\ge {z}\ge  \beta/{h }
\end{align}
To bound the ratio $\ell^N(R)/M(\pi^*)$ from above
we will solve the general problem (\ref{eq:optim0})-(\ref{eq:optim3a}), where  we ignore our definitions of $\beta$ and $z$ above and thus allow $\beta$ and
$z$ to take any nonnegative real values (subject to (\ref{eq:optim3a})).

Since $\frac{h}i+\frac{i}h \ge2$ for all $i\ge1$ and $h\ge1$,
we see that for any $\beta\ge0$ and $h\ge1$,
the denominator in (\ref{eq:optim0}) is
positive for every feasible solution $(\{D_i\}_{i=1}^h,{z})$ of (\ref{eq:optim1a})--(\ref{eq:optim3a}).
We can therefore also consider the following equivalent \emph{minimization} problem:
\begin{equation}
\label{eq:minimize2}
 \min\left\{\left.
\sum_{i=1}^{h} \left(\frac{h}{i}+\frac{i}h\right) D_i - {z}
\right|
(\ref{eq:optim1a})\mbox{--}(\ref{eq:optim3a}) 
\right\}\,.
\end{equation}

In what follows, we solve (\ref{eq:minimize2}) for any fixed $h$ and $\beta$, and then determine
which values of $h$ and $\beta$ give the highest overall value for (\ref{eq:optim0}).
For fixed $h$ and $\beta$,
any solution $(\{D_i\}_{i=1}^h,{z})$ of (\ref{eq:optim1a}) -- (\ref{eq:optim3a}) is either an optimal solution to
both problem (\ref{eq:optim0})--(\ref{eq:optim3a}) and problem (\ref{eq:minimize2}) or to neither of them.
The next lemma helps to simplify our problem (\ref{eq:minimize2}), and hence problem (\ref{eq:optim0})--(\ref{eq:optim3a}).

\begin{lemma}
\label{lem:di}\label{adj}
There is an optimal solution $(\{D_i\}_{i=1}^h,{z})$
of (\ref{eq:minimize2}), which is also an optimal solution of  (\ref{eq:optim0}) -- (\ref{eq:optim3a}), such that $D_i>0$ for at most two values of $i$. If there are two such values, they are consecutive.
\end{lemma}
\begin{proof}
Consider {an optimal} solution $(\{D_i\}_{i=1}^h,{z})$ of (\ref{eq:minimize2}).
Suppose for a contradiction that there exist two indices $i_1\leq i_2-2$ such that
$D_{i_1}>0$ and $D_{i_2}>0$.
We can modify the solution $\{D_i\}_{i=1}^h$ as follows.
Let $m_j$ ($j=1,2$) be real values with $0<m_j\le D_{i_j}$.
Subtract $m_1$ from $D_{i_1}$ and add it to $D_{i_1+1}$.
Subtract $m_2$ from $D_{i_2}$ and add it to $D_{i_2-1}$.
Then we still have  $\sum_{i=1}^h D_i=1$ and $D_i\ge0$ $(i=1,\dots,h)$.

We need to determine $m_1$ and $m_2$ so that constraint (\ref{eq:optim1a}) is still satisfied.
To this end, we investigate by how much the sum on the left-hand side of (\ref{eq:optim1a}) increases, i.e., 
\begin{align*}
\left(-\frac{2i_1}{h}+\frac1{i_1}+\frac{2(i_1+1)}{h}-\frac1{i_1+1}\right)m_1
+\left(\frac{2(i_2-1)}{h}-\frac1{i_2-2}-\frac{2i_2}{h}+\frac1{i_2}\right)m_2\\
= \left(\frac{1}{i_1}+\frac{2}{h}-\frac1{i_1+1}\right)m_1
-\left(\frac{2}{h}+\frac1{i_2-1}-\frac1{i_2}\right)m_2,
\end{align*}
which should be at most 0 in order to maintain a feasible solution. This is equivalent to requiring that $\frac{m_1}{m_2}$ be bounded from above by
\begin{equation}
\label{eq:m1m2}
  \left({\frac1{(i_2-1)i_2}+\frac2h}\right)
\left/\middle(                         {\frac1{(i_1+1)i_1}+\frac2h}\right) =:\alpha_{i_1i_2}\leq1,
\end{equation}
where the last inequality holds since $i_2\geq i_1+2$.
On the other hand, we aim at decreasing the objective function in (\ref{eq:minimize2}) with this procedure. Therefore, we require the increase of the objective value to be negative. From this we get
 \begin{align*}
 0
 & >
\left(-\frac{h}{i_1}-\frac{i_1}{h}
+\frac{h}{i_1+1}+\frac{i_1+1}{h}\right)m_1
+\left(\frac{h}{i_2-1}+\frac{i_2-1}{h}
-\frac{h}{i_2}-\frac{i_2}{h}\right)m_2\\
& = \left(\frac{-h}{i_1(i_1+1)}+\frac{1}{h}\right)m_1
        + \left(\frac{h}{(i_2-1)i_2}-\frac1h\right)m_2\\
\Rightarrow& \left(\frac{h}{(i_2-1)i_2}-\frac1h\right)m_2
<
\left(\frac{h}{i_1(i_1+1)}-\frac{1}{h}\right)m_1.
\end{align*}
Note that the coefficients of $m_1$ and $m_2$ are positive since $i_1+2\leq i_2\leq h$.
Therefore, requiring that the increase of the objective function be negative is equivalent to requiring that $\frac{m_1}{m_2}$ be greater than
\begin{equation}
\label{eq:b12}
\left({\frac{h}{(i_2-1)i_2}-\frac1h}\right)
\left/\middle({\frac{h}{i_1(i_1+1)}-\frac{1}{h}}\right)
=: \beta_{i_1i_2}
\end{equation}
From the definitions in (\ref{eq:m1m2}) and (\ref{eq:b12}),
doing crosswise multiplication, it is easy to check that $\beta_{i_1i_2} < \alpha_{i_1i_2}$ for $i_1+2\leq i_2\leq h$.
This shows that there exist positive values $m_1$ and $m_2$ such that
$\alpha_{i_1i_2} \geq  \frac{m_1}{m_2} > \beta_{i_1i_2}, \text{ and }m_j\leq D_{i_j}\text{ for }j=1,2.$
Thus, as a result of our modification of
the sequence $\{D_i\}_{i=1}^h$, the
objective function value in (\ref{eq:minimize2}) decreases  by a positive amount, and the constraints {(\ref{eq:optim1a})\mbox{--}(\ref{eq:optim3a})} are still satisfied.
This contradicts the optimality of $(\{D_i\}_{i=1}^h,{z})$.
\end{proof}
For an optimal solution $(\{D_i\}_{i=1}^h,{z})$ to   (\ref{eq:optim0})--(\ref{eq:optim3a})
as given in Lemma \ref{adj}, 
let $x\in[h-1]$ be the minimum index such that $D_i=0$ 
for all $i\in[h]\backslash\{{x},{x}+1\}$.
That is, $x$ is the minimum index such that $D_x>0$, or $x=h-1$.
Writing ${y}$ for $D_{{x}+1}$, we have $D_{x}=1-{y}$, and problem (\ref{eq:optim0}) -- (\ref{eq:optim3a}) transforms  to the following relaxation which drops the upper bound $\beta$ on ${z}$ in (\ref{eq:optim3a}).

\begin{align}
\label{eq:optimx}
 \frac{\ell^N(R)}{M(\pi^*)}\le\max &
\frac{ 1+\beta}
     { \frac{h}{x}+\frac{x}{h}-\left(\frac{h}{x(x+1)}-\frac{1}{h}\right) y  -2+\beta-z}\\
\label{eq:optimx2}
\text{s.t.} &  \frac{2x}h-\frac{1}{x}+\left(\frac{2}{h}+\frac{1}{x(x+1)}\right) y  +2z
	 \leq \frac{h-1}{h}+\beta\\
\label{eq:optimx2d}
& 1\le x \le h-1, x\in\mathbb N\\
\label{eq:optimx3}
& 0\le y \le1\\
\label{eq:optimx3b}
& \beta/{h}\le z
\end{align}

For convenience, we restate the problem (\ref{eq:optimx}) --
(\ref{eq:optimx3b}) as follows \be \label{eq:optimx1}
\frac{\ell^N(R)}{M(\pi^*)}\leq  \max\left\{ \frac{1 +
\beta}{f({x},{y},{z})}\, \left|\,g({x},{y},{z})  \le  0,\;  {x}\in
\mathbb N,\;{y}  \in  [0,1],\;{z}  \geq   \frac{\beta}h
\right.\right\}, \ee where we have relaxed the constraint on $x$ and
\begin{align*}
 f({x},{y},{z})&=\frac{h}{{x}}+\frac{{x}}{h}
-\left(\frac{h}{{x}({x}+1)}-\frac{1}{h}\right) {y}
-{z} - 2 + \beta,\\
g({x},{y},{z})&=
\frac{2x}h-\frac{1}{{x}}
+\left(\frac{2}{h}+\frac{1}{{x}({x}+1)} \right) {y} +2{z}-\frac{h-1}h-\beta.\end{align*}
We turn to
consider the corresponding minimization of $ f({x},{y},{z})$ under the same constraints. It is clear that
 the optimal value of the minimization is attained at $g({x},{y},{z}) =0$ (if $g({x},{y},{z}) <0$, we can increase $z$ and decrease the objective function). So we only need to consider the minimization problem, as well as its relaxation, with this equality constraint:
\begin{align}
\Omega_1 &:= \min\left\{   f({x},{y},{z}) \,
\left|\,g({x},{y},{z})=0, \; {x}\in  \mathbb N,\;{y}  \in  [0,1],\;{z}  \geq \beta/h \right.\right\}\label{o1}\\
\Omega_2 &:= \min\left\{   f({x},{y},{z}) \,
\left|\,g({x},{y},{z}) = 0,\;  {x}\ge1,\;{y}  \in  [0,1],\;{z} \geq \beta/h\right.\right\}\label{o2}\end{align}
\begin{observation}\label{om}
For $i=1$ or $2$, if $\Omega_i\ge\frac{1+\beta}3$, then $ \frac{\ell^N(R)}{M(\pi^*)}\leq3$.
\end{observation}

\paragraph{Main ideas} In view of Observation \ref{om},
  we will prove    $\Omega_1\ge\frac{1+\beta}3$  for $h\in\{3,4,6\}$ in Section \ref{sec346}, and $\Omega_2\ge\frac{1+\beta}3$ for $h\geq7$ in Section \ref{large}. The proof for $\Omega_1$ utilizes a case analysis, which is simplified by the fact that every optimal solution of (\ref{o1}) when $h\in\{3,4,6\}$ has its $y$ or $z$ touch the boundary.
  The key idea for lower bounding $\Omega_2$ is using the fact that the optimal solution of (\ref{o2}) must
be a KKT point (a solution satisfying the Karush-Kuhn-Tucker (KKT) conditions).  We will bound the values of objective function $f$ at all KKT points of (\ref{o2}) from below by $\frac{1+\beta}3$.

To lower bound $\Omega_1$ and $\Omega_2$, we need to consider the derivatives of the objective and constraint functions.
 Using ${x}\ge1$ and ${y}\in[0,1]$, we obtain
\begin{align*}
\frac{\partial f}{\partial {x}}=&\frac{1}{h}-\frac{h(1-{y})}{{x}^{2}}-\frac{h{y} }{({x}+1)^{2}},
& \frac{\partial f}{\partial {y} }=&\frac{1}{h}-\frac{h}{{x}({x}+1)}<0,
& \frac{\partial f}{\partial {z}}=&-1,\\
\frac{\partial g}{\partial {x}}=&\frac{2}{h}+\frac{1-{y} }{{x}^{2}}+\frac{{y} }{({x}+1)^{2}}>0,
& \frac{\partial g}{\partial {y} }=&\frac{2}{h}+\frac{1}{{x}({x}+1)}>0,
& \frac{\partial g}{\partial {z}}=&2.
\end{align*}
For brevity we define
$ \chi:=( \sqrt{2h^{2}-h+1}-1)/{2}\text{ and }\nu:= \sqrt{2h^{2}-h}\,/{2}.$
It is straightforward to verify the following equivalences.
\begin{lemma}
 \label{lem:9}
\begin{enumerate}[(i)]
\item
 $\frac{\partial f/\partial {x} }{\partial g/\partial {x} }
= \frac{\partial f/\partial {y}}{\partial g/\partial {y}}\Leftrightarrow{y} = \frac{{x}+1}{2{x}+1}
\Leftrightarrow\frac{\partial f}{\partial x}=\frac{\partial f}{\partial y}
\Leftrightarrow\frac{\partial g}{\partial x}=\frac{\partial g}{\partial y}$.
\item
 $\frac{\partial f/\partial {y} }{\partial g/\partial {y} }\ge\frac{\partial f/\partial{{z}}}{\partial g/\partial{{z}}}\Leftrightarrow
 x \ge \chi$, and  $\frac{\partial f/\partial {y} }{\partial g/\partial {y} }=\frac{\partial f/\partial{{z}}}{\partial g/\partial{{z}}}\Leftrightarrow
x =\chi$.
\item
 If ${y}\in\{0,1\}$, then 
 $\frac{\partial f/\partial {x}}{\partial g/\partial {x}}=\frac{\partial f/\partial{{z}}}{\partial g/\partial{{z}}}\Leftrightarrow {x}+{y}=\nu$.
\item
 If ${y}=\frac{{x}+1}{2{x}+1}$, then $\frac{\partial f/\partial {x}}{\partial g/\partial {x}}=\frac{\partial f/\partial{{z}}}{\partial g/\partial{{z}}}\Leftrightarrow x=\chi$.
\end{enumerate}
\end{lemma}

\begin{lemma}\label{eo}
Let $({x}^*,{y}^{*},{z}^*)$ be an optimal solution to   (\ref{o1}) or  (\ref{o2}).
\begin{enumerate}[(i)]
\item If ${x}^*<\chi$, then ${y}^{*}=1$ or ${z}^*=\beta/h$.
\item If ${x}^*>\chi$, then  ${y}^{*}=0$.
\item $x^*\le h-1$.
\end{enumerate}
\end{lemma}
\begin{proof}
(i)--(ii)
If  ${x}^*<\chi$ (resp. ${x}^*>\chi$), then it follows from Lemma \ref{lem:9}(ii) that $\frac{\partial f/\partial {y} }{\partial g/\partial {y} }<\frac{\partial f/\partial{{z}}}{\partial g/\partial{{z}}}$ (resp. $\frac{\partial f/\partial {y} }{\partial g/\partial {y} }>\frac{\partial f/\partial{{z}}}{\partial g/\partial{{z}}}$).  Since increasing ${y}^{*}$ and decreasing ${z}^*$ (resp. increasing ${z}^*$ and decreasing ${y}^{*}$) cannot give a better solution for the problem, it must be the case that ${y}^{*}=1$ or ${z}^*=\beta/h$  (resp. ${y}^{*}=0$).

(iii) We have $\frac{\partial f}{\partial {x}}=0$ only if $x=h-1$,
and $\frac{\partial f}{\partial {x}}<0$ only if $1\le x<h-1$. Therefore, if $x^*>h-1$, decreasing $x^*$ would give smaller objective value.
\end{proof}

\subsection{Covering Equilibria with  $h\in\{3,4,6\}$}\label{sec346}
In this case we lower bound   $\Omega_1 $ by $(1+\beta)/3$, which
along with Observation \ref{om} implies the upper bound of 3 on
$\ell^N(R)/M(\pi^*)$ for $h\in\{3,4,6\}$.

\begin{lemma}\label{smallh}
For $h\in\{3,4,6\}$, $\Omega_1\ge\frac{1+\beta}3$.
\end{lemma}
\begin{proof} Let $({x}^*,{y}^*,{z}^*)$ denote an optimal solution to
problem (\ref{o1}). Notice that ${x}^*\in[h-1]$ is an integer, which
cannot be equal to the noninteger $\chi$ for any $h\in\{3,4,6\}$. By
Lemma \ref{eo}(i)--(ii), we have  ${z}^*=\beta/h$ or ${y}^*=1$ if
${x}^*<\chi$ and ${y}^*=0$ otherwise. We calculate and estimate
$\Omega_1=f({x}^*,{y}^*,{z}^*)$ in Tables \ref{table:1} and
\ref{table:2} below by checking all necessary ${x}^*$, using Lemma
\ref{eo}(iii) (see the fourth column of Table  \ref{table:1} and the
third column of Table  \ref{table:2}). Table  \ref{table:1} presents
the cases for ${x}^*<\chi$ and ${z}^*=\beta/h$, where
${y}^*\in[0,1]$ is determined by $g({x}^*,{y}^{*},{z}^*)=0$, and
Table \ref{table:2} presents the cases for ${y}^*\in\{0,1\}$. In all
cases we obtain the claimed lower bound.
\end{proof}
\subsection{Covering Equilibria with  $h\ge7$}\label{large}
In this section, we assume $h\ge7$. Our goal is to prove the
optimal objective value $\Omega_2$ of problem (\ref{o2}) is at most
$ ({1+\beta})/3$ for all $h\ge7$. Throughout Section \ref{large}, we
assume $(x^*,y^*,z^*)$ to be a {\em fixed optimal solution} to
(\ref{o2}) such that $y^*$ is minimum. In particular, since
$f(x,1,z)=f(x+1,0,z)$ and $g(x,1,z)=f(x+1,0,z)$ for all
$x>0,z\in\mathbb R$, we can assume without loss of generality that
$y^*<1$. In the following, we distinguish among three cases:
\begin{enumerate}[1)]
\item ${x}^*=1$ (Claim \ref{case1}),
\item ${x}^*>1$ and ${y}^*=0$  (Claim \ref{case2}), and
\item ${x}^*>1$ and $0<y^*<1$ (Claim \ref{case3}).
\end{enumerate}

\begin{table}[t]
\begin{center}
{
\renewcommand{\arraystretch}{1.5}
\begin{tabular}{c| c  c||  c|| r c| l}
\hline
\mbox{}\ \ $ h$\mbox{}\ \  & \mbox{}\ \ ${z}^*$\mbox{}\ \  &\mbox{}\ \ $\chi$\mbox{}\ \  & \mbox{}\ \ ${x}^* $ \mbox{}\ \   
&\mbox{}\ \  $g({x}^* ,{y}^*  ,{z}^* ) = 0$    &${y}^* $ &\mbox{}\ \ $\mbox{}\ \  \Omega_1=f({x}^*,{y}^*,{z}^*) $ \\  
 \hline

3& $\frac{\beta}3$ &1.5  & 1 & $\frac{7}6{y}^*+2{z}^*-1-\beta=0$& $\frac{6}{7}+\frac{2\beta}{7}$ &\mbox{}\ \ $\frac{1+\beta}{3}$ \\
\hline

& & & 1  & ${y}^*+2{z}^*-\frac54-\beta=0$ &$\not\in[0,1]$    &
\mbox{}\ \  infeasible \\

\raisebox{1ex}{4} & \raisebox{1ex}{$\frac{\beta}4$} &
\raisebox{1.3ex}{2.19}
& 2    &$\frac23{y}^*+2{z}^*-\frac14-\beta=0$ &$\frac38+\frac34\beta$   &\mbox{}\ \ $ \frac{75-18\beta }{32}>\frac{1+\beta}{3}$ \\
\hline

& & & 1    &$\frac{5 }{6}{y}^*+2{z}^*-\frac32-\beta=0$
& $\not\in[0,1]$ &\mbox{}\ \  infeasible\\

 6&$\frac{ \beta}6$ &3.59  & 2    &$\frac{1 }{2}{y}^*+2{z}^* -\frac23-\beta=0$
 & $\not\in[0,1]$    & \mbox{}\ \ infeasible \\

 &  & & 3  &\mbox{}\ \ $\frac{5 }{12}{y}^*+2{z}^*-\frac16-\beta=0$\mbox{}\ \
 &\mbox{}\ \  $\frac25+\frac{8\beta}5$ \mbox{}\ \   &\mbox{}\ \ $\frac{71-21\beta }{30}>\frac{1+\beta}{3}{}$\footnotemark \\  
\hline 
\end{tabular}
}
\vspace{3mm}\caption{\label{table:1} $\Omega_1\ge\frac{1+\beta}3$ when ${x}^*<\chi$ and ${z}^*=\frac{\beta}h$ for $h=3,4,6$. }
\end{center}
\end{table}
\footnotetext{Note that $\beta\le3/8$.}

 \begin{table}[t]
 \begin{center}
{
\renewcommand{\arraystretch}{1.5}
\begin{tabular}{c| c ||  c|| c c| l}
\hline
\mbox{}\ \ $ h$\mbox{}\ \  &\mbox{}\ \  $\chi $\mbox{}\ \  &\mbox{}\ \  $  {x}^* $\mbox{}\ \ 
&\mbox{}\ \  ${y}^*$  \mbox{}\ \    &\mbox{}\ \ ${z}^* $\mbox{}\ \
&\mbox{}\ \
$\Omega_1= f({x}^*,{y}^*,{z}^*)$ \mbox{}\ \ \\

 \hline
3 & 1.5 &  1, 2 &\mbox{}\ \ $2-x^*$&   $\frac{\beta}2-\frac1{12}$\mbox{}\ \  &\mbox{}\ \  {$\frac{9-2\beta }{4}\ge\frac{1+\beta}{3}$\footnotemark}\\
\hline
& & 1 &$1 $ &$  \frac{\beta}2+\frac18\mbox{}\ \   $&\mbox{}\ \  $\frac{19-4\beta }{8}\ge\frac{1+\beta}{3}$\\

 \raisebox{1ex}{4} & \mbox{}\ \ \raisebox{1ex}{2.19}\mbox{}\ \  & 2, 3 &\mbox{}\ \ $3-x^*$ & $ \frac{\beta}2-\frac5{24} $ \mbox{}\ \ & \mbox{}\ \ $ \frac{55-12\beta }{24}>\frac{1+\beta}{3}$\footnotemark\\
\hline

   &   & 1&$ 1$ &$ \frac{\beta}2+\frac13\mbox{}\ \  $&\mbox{}\ \ $\frac{6-\beta }{2}>\frac{1+\beta}{3}$\\

   &   & 2 &$ 1  $ &$ \frac{\beta}2+\frac1{12}\mbox{}\ \  $& \mbox{}\ \ $\frac{29-6\beta }{12}>\frac{1+\beta}{3}$\\

 \raisebox{1ex}{6} & \mbox{}\ \  \raisebox{1ex}{3.59}\mbox{}\ \ & 3, 4  & \mbox{}\ \  $4-x^*$ & $ \frac{\beta}2-\frac18 $\mbox{}\ \ &
 $\mbox{}\ \  \frac{55-12\beta }{24}>\frac{1+\beta}{3}$\footnotemark\\

    &  & 5 &$ 0 $ &$ \frac{\beta}2-\frac{19}{60} $\mbox{}\ \ &\mbox{}\ \  $\frac{47-10\beta }{20}>\frac{1+\beta}{3}$\\
\hline 
\end{tabular}
}
\vspace{3mm}\caption{\label{table:2} $\Omega_1\ge\frac{1+\beta}3$ when ${x}^*<\chi$  and ${y}^*=1$ (resp.  ${x}^*>\chi$ and ${y}^*=0$) for $h=3,4,6$.}
\end{center}
\end{table}

\addtocounter{footnote}{-2}
\footnotetext{Note that $\beta\ge1/2$.}
\stepcounter{footnote}
\footnotetext{Note that $\beta\ge5/6$.}
\stepcounter{footnote}
\footnotetext{Note that $\beta\ge3/8$. }

Our basic tool is the KKT  conditions which $(x^*,y^*,z^*)$ must satisfy. For notational convenience, we also express the  constraints $x\ge1,y\in[0,1],z\ge\beta/h$ as $g_i(x,y,z)\le0$, $i=1,2,3,4$, respectively, where $g_1(x,y,z)=-x+1$, $g_2(x,y,z)=y-1$, $g_3(x,y,z)=-y$ and $g_4(x,y,z)=-z-\beta/h$. By the KKT conditions on the minimization problem (\ref{o2}), there exist constant $\lambda$ and nonnegative constants $\mu_i$ ($1\le i\le 4)$ such that
\begin{align}\label{kt}
\nabla f(x^*,y^*,z^*)+\lambda\nabla g(x^*,y^*,z^*)+\sum_{i=1}^4\mu_i\nabla g_i(x^*,y^*,z^*)&=\bfm0\\
\label{kt2}
\mu_ig_i( x^*,y^*,z^*)&=0\text{ for }i\in[4].
\end{align}

\paragraph{\bf Case 1:} ${x}^*=1$.\quad This case is handled by the following claim.
\begin{claim}\label{case1}
$\min\left\{   f({x},{y},{z}) \,
\left|\,g({x},{y},{z}) = 0,\,{x}=1,\,{y}  \in  [0,1] \right.\right\}>(1+\beta)/3$.
\end{claim}
\begin{proof}
If $x=1$, we have
${z} =1+\frac{\beta}2-\frac3{2h}-(\frac1h+\frac14){y} $ from the constraint $g({x} ,{y} ,{z} )=0$. It follows that $f({x} ,{y} ,{z} ) =  {h+\frac5{2h}-(\frac{h}2-\frac14-\frac2h){y} -3+\frac{\beta}2}$,
which together with $h\ge7$, ${y} \le1$ and $\beta\ge0$ implies
$f({x} ,{y} ,{z} ) \ge f(1,1,z)=\frac h2+\frac9{2h}-\frac{11}{4}+\frac\beta2
 >\frac72-\frac{11}4+\frac\beta2>
\frac{1+\beta}3$.
\end{proof}

Before turning to the next two cases, we prove a few technical lemmas.
\begin{claim}\label{rm}
If $h\ge7$ and $1 +\beta-\frac{1+2\beta}{h}>\frac{\sqrt{2h^2-h}}{h}-\frac{2}{\sqrt{2h^2-h}}$, then $\beta>\sqrt2-1$.
\end{claim}
\begin{proof}
{The function $1 +\beta-\frac{1+2\beta}{h}$ increases in $\beta$ for
$h>2$, and
$\sqrt{2}-\frac{2\sqrt{2}-1}{h}\le\frac{\sqrt{2h^2-h}}{h}-\frac{2}{\sqrt{2h^2-h}}$
for $h\ge7$.}
\end{proof}

In the next two lemmas, we write {$\beta_{-1}:=\beta-1$}, $\beta_1:=\beta+1$ and $\beta_2:= 2\beta+1$ for brevity.

\begin{lemma}\label{lema}
For any constant $\beta\in[0,\sqrt2-1)$, if variables $\hbar$ and
$x$ satisfy {$\hbar\ge7$} and
$x=\frac14(\beta_1{\hbar}-\beta_2+\sqrt{(\beta_1{\hbar}-\beta_2)^{2}+8{\hbar}})$,
then $\Psi({\hbar}):=\frac {\hbar} x +\frac x {\hbar} - \frac \beta
{\hbar}\ge\Psi(7)$.
\end{lemma}
\begin{proof}Notice that
${\hbar}=\frac{2x^2+\beta_2x}{\beta_1x+1}$.  Thus $\Psi({\hbar})$
can be considered as the function of $x$, which we write as
$\psi(x)$.
\[\psi(x)=\frac{2x+\beta_2}{\beta_1x+1}
+\frac{\beta_1x-\beta_1\beta+1}{2x+\beta_2}
-\frac{\beta}{2x^2+\beta_2x}=\Psi({\hbar}).\]
It is easy to check that $x$ is monotonically increasing in ${\hbar}$;  in particular ${\hbar}\ge7$ along with $\beta\ge0$ implies $x\ge3.897$.

Moreover, to prove the lemma, we only need to show that $\psi(x)$ is monotonically increasing in $x$.  Observe that the last term in the derivative of the above expression
\[
\frac{d\psi}{dx}=\frac{1-3\beta-2\beta^2}{(\beta_1x+1)^2}
+\frac{4\beta^2+5\beta-1}{(2x+\beta_2)^2}
+\frac{\beta(4x+\beta_2)}{(2x^2+\beta_2x)^2}
\]
is positive. It suffices to verify
\be\label{goal}
 \frac{1-3\beta-2\beta^2}{(\beta_1x+1)^2}+\frac{4\beta^2+5\beta-1}{(2x+\beta_2)^2} \ge 0.
\ee

In case of $0\le\beta\le\frac14(\sqrt{17}-3)<0.2808$, we obtain
$1-3\beta-2\beta^2\ge\max\{0,1-5\beta-4\beta^2\}$, and (\ref{goal})
is true as $0<\frac{(\beta_1x+1)^2}{(2x+\beta_2)^2}\leq1$ for
$0\le\beta\le\sqrt2-1$.

In case of $\beta\in
(\frac14(\sqrt{17}-3),\sqrt2-1)\subset[0.28,\sqrt2-1)$, we have
$2\beta^2+3\beta-1>0$. It can be seen from $x\ge3.897$ that
$(1.7\beta-0.3)x>2\beta-0.7$, implying
$\frac{2x+\beta_2}{\beta_1x+1}\le {1.7}\,. $ On the other hand,
(\ref{goal}) follows from
\[
\frac{4\beta^2+5\beta-1}{2\beta^2+3\beta-1}>2+\frac{1-(\sqrt2-1)}{2(\sqrt2-1)^2+3(\sqrt2-1)-1}=3>1.7^2
\ge\frac{(2x+\beta_2)^2}{(\beta_1x+1)^2}.
\]
 The lemma is proved.
\end{proof}

\begin{lemma}\label{lemb}
For any constant $\beta\in[0,0.13)$, {if variables $\hbar$ and $x$ satisfy 
${\hbar}=\frac{(2x+1)^2+\beta_2(2x+1)+1}{\beta_1(2x+1)+2}\ge7$, then function $\Psi({x})\equiv\frac{{\hbar}}{{x}}+\frac{{x}}{{\hbar}}-\left(\frac{{\hbar}}{{x}({x}+1)}-\frac1{\hbar}\right)\frac{{x}+1}{2{x}+1}
-\frac{\beta}{\hbar}>2.36$.}
\end{lemma}
\begin{proof}   Let $u=2{x}+1$. Then {${\hbar}=\frac{u^2+\beta_2u+1}{\beta_1u+2}\ge7$, and $u=\frac{\beta_1{\hbar}-\beta_2+\sqrt{(\beta_1{\hbar}-\beta_2)^2+8{\hbar}-4}}2$ is lower bounded by $ \frac{ {\hbar}-1 +\sqrt{({\hbar}-1)^2+8{\hbar}-4}}2\ge\frac{6+\sqrt{88}}2>7.69$, as ${\hbar}\ge 7$} and $\beta\ge0$.

It is routine to check that {$\Psi({x})  =\frac{2{\hbar}}{u}+\frac{u}{2{\hbar}}+\frac1{2{\hbar}u}-\frac{\beta}{\hbar}$, and it is a   function $\psi$ of $u$ with derivative $d\psi/du$ as follows:
\begin{align*}
\psi(u): =&\frac{\beta_1u^3+(2-2\beta\beta_1)u^2+(1-3\beta)u+2}{2u(u^2+\beta_2u+1)}+\frac{2(u^2+\beta_2u+1)}{u(2+\beta_1u)}
=\Psi({x}). \\
 \frac{d\psi}{du} =&\frac{(4\beta_1\beta+\beta_{-1})u^4+8\beta u^3+(4\beta_{-1}\beta_1-\beta_1)u^2-4\beta_2u-2}{2u^2(u^2+\beta_2u+1)^2}\\
&+\frac{2(1-3\beta-2\beta^2)u^2-4\beta_1u-4}{u^2(2+\beta_1u)^2}
\end{align*}
By $u\ge7.69$ and $\beta\in[0,0.13)$, it is easy to see that} the numerator of
the second term in the above expression of $d\psi/du$ is positive.   Since $\sqrt{2}(u^2+\beta_2u+1)>u(2+\beta_1u)$ holds
for any $\beta\in[0,0.13)$, we have
\begin{align*}
\frac{d\psi}{du} >&\frac{(4\beta_1\beta+\beta_{-1})u^4+8\beta u^3+(4\beta_{-1}\beta_1-\beta_1)u^2-4\beta_2u-2}{2u^2(u^2+\beta_2u+1)^2}
\\ & +\frac{2(1-3\beta-2\beta^2)u^4-4\beta_1u^3-4u^2}{2u^2(u^2+\beta_2u+1)^2} \\
 =&\frac{(1-\beta)u^4-4(1-\beta)u^3-(9+\beta-4\beta^2)u^2-4(2\beta+1)u-2}{2u^2(u^2+\beta_2u+1)^2} \\
 \geq&\frac{0.87u^4-4u^3-9.0624u^2-5.04u-2}{2u^2(u^2+\beta_2u+1)^2}
\end{align*}
The numerator is positive as $u>7.69$. Therefore $d\psi/du>0$.

{Using  $u>7.69$ and $\beta<0.13$, we obtain
$\Psi({x})=\psi(u)\ge\psi(7.69)\ge
\frac{2{\hbar}}{7.69}+\frac{7.69}{2{\hbar}}+\frac1{15.38{\hbar}}-\frac{0.13}{\hbar}=\frac{200{\hbar}}{769}+\frac{581367}{153800{\hbar}}$,
which increases in ${\hbar}$ for all ${\hbar}\ge7$.} Thus
$\Psi(x)\ge
\frac{200\times7}{769}+\frac{581367}{153800\times7}>2.36$.~
\end{proof}

\paragraph{\bf Case 2:} ${x}^*>1$ and ${y}^*=0$.\quad
{We begin by considering the following relaxed problem, which does not have a bound on $z$.}
\bea
\Omega_3:= \min\left\{   f({x},y,{z}) \,
\left|\,g({x},y,{z}) = 0,\;  {x}\ge1, y=0\right.\right\}.
\label{relax}
\eea
{Clearly $\Omega_3\le \Omega_2$.}
The KKT conditions applied to (\ref{relax}) assert that $\Omega_3$ is attained at some feasible solution $(x,z)$ of (\ref{relax}) for which there exist constants $\theta$ and $\eta$ such that
\begin{align*}
\nabla f(x,0,z)+\theta\nabla g(x,0,z)+\eta \nabla(-x+1) &=\bfm 0\\
\eta(-x+1)&=0
\end{align*}
It follows that  $\Omega_3$ is attained
either when ${x}=1$ or when ${x}>1\Rightarrow\eta=0\Rightarrow\frac{\partial f/\partial {x}}{\partial g/\partial {x}}=-\theta
= \frac{\partial f/\partial {z}}{\partial g/\partial {z}}$ holds at $(x,0,z)$. In the former case, we are done by Claim \ref{case1}. In the latter case, Claim \ref{lem:9}(iii) gives {$x=\nu$, and therefore $g({x},0,{z}) = g(\nu ,0,{z})=0$}
implies
 \[{z}= \left.\left(\frac{h-1-2\nu}{h} +\frac{1}{\nu}+\beta\right)\right/2=:z_2.\]
Notice that $\Omega _3=f(\nu,0,{z_2})=\frac{4\sqrt{2h^2-h} +1}{2h}+\frac{\beta-5}2$
increases in $h$ for $h\ge7$, and hence $\Omega_3\ge
\frac{4\sqrt{91}+1}{14}+\frac{\beta-5}2 $, which is greater than $
\frac{1+\beta}3$ if $\beta\ge\sqrt2-1$.  This together with
$\Omega_3\le\Omega_2$ verifies the following
 \be\label{sqrt2}
 \Omega_2>(1+\beta)/3\,\text{ if }\,\beta\ge\sqrt2-1.
 \ee

  We next turn back to  (\ref{o2}), and investigate its optimal solution $(x^*,0,z^*)$.
   \begin{claim}\label{case2}
If ${x}^*>1$ and ${y}^*=0$, then $\Omega_2\ge(1+\beta)/3$.
\end{claim}
\begin{proof}
If ${z}^*>\beta/h$, since ${x}^*>1$, the KKT conditions (\ref{kt})--(\ref{kt2}) imply that $\mu_1=\mu_4=0$ and $\frac{\partial f/\partial {x}}{\partial g/\partial {x}}
= -\lambda=\frac{\partial f/\partial {z}}{\partial g/\partial {z}}$ holds at $({x}^*,0,{z}^*)$. Using ${y}^*=0$ and Lemma \ref{lem:9}(iii), we obtain  $x^*=\nu$. In turn $g({x}^*,0,{z}^*)=0$ gives ${z}^*={z_2}$. So ${z_2}>\beta/h$, which reads
$1 +\beta-\frac{1+2\beta}{h}>\frac{\sqrt{2h^2-h}}{h}-\frac{2}{\sqrt{2h^2-h}}$. It follows from Lemma \ref{rm}  that  $\beta>\sqrt{2}-1$, and further from (\ref{sqrt2}) that $\Omega_2>(1+\beta)/3$.

If ${z}^*=\beta/h$, by (\ref{sqrt2}), we only need to consider the case where $\beta\in[0,\sqrt2-1)$. The  constraint {$g({x}^*,0,{z}^*)=\frac{2x^*}h-\frac1{x^*}+2{z}^*-\frac{h-1}h-\beta$ gives $$
x^*= \frac14\left((\beta+1) h-(2\beta+1)+\sqrt{((\beta+1) h-(2\beta+1))^{2}+8h}\right):= x^*(h).$$}
It follows from Lemma \ref{lema} that {$f({x}^*,0,{z}^*)=\frac{h}{x^*}+\frac{x^*}h-\frac{\beta}h-2+\beta\ge\frac{7}{x^*(7)}+\frac{x^*(7)}7+\frac{6\beta}7-2$.} This value is easily checked to be  $\frac{15}{28}\sqrt{(6+5\beta)^2+56}-\frac{41}{28}\beta-\frac{67}{14}$, which is smaller than $(1+\beta)/3$ for $\beta\ge0$.
\end{proof}

\paragraph{\bf Case 3:} ${x}^*>1$ and $0<y^*<1$.\quad
 In this  case, the KKT conditions   (\ref{kt})--(\ref{kt2})   imply $\mu_i=0$ for $1\le i\le 3$ and $\frac{\partial f/\partial {y} }{\partial g/\partial {y} }=-\lambda= \frac{\partial f/\partial {x}}{\partial g/\partial {x}}$ holds at  $({x}^*, {y}^*, {z}^*)$. In turn,
Lemma \ref{lem:9}(i) asserts ${y}^*=\frac{{x}^*+1}{2{x}^*+1}$, implying
\be
\Omega_2\ge\Omega_4:=\min\left\{f({x},{y},{z})\,
\left|\,g({x},{y},{z})=0,  {x}\ge1,{y}=\frac{{x}+1}{2{x}+1} \right.\right\}.
\label{o4}
\ee
 From the KKT conditions on the minimization (\ref{o4}), we deduce that $\Omega_4$ is attained at some feasible solution $(x,y,z)$ of (\ref{o4}) for which there exist constants $\theta_1,\theta_2,\eta_1,\eta_2$ such that
 \begin{align}
\label{kkt3}
\nabla f(x,y,z)+\theta_1\nabla g(x,y,z)+\theta_2\nabla(y-\frac{x+1}{2x+1})+\eta\nabla(-x+1) &=\bfm 0\\
\label{kkt3a}
\eta(-x+1)&=0
  \end{align}
  It follows that $\Omega_4$
 is attained  either when ${x}=1$ or when ${x}>1$, in which case $\eta=0$ by (\ref{kkt3a}).  In the former case, we are again done by Claim \ref{case1}. In the latter case, from (\ref{kkt3}) we find
\begin{align}
\label{kkt3-1}
\frac{\partial f}{\partial {x}}+\theta_1\frac{\partial g}{\partial {x}}&=-\frac{\theta_2}{(2x+1)^2}\\ \label{kkt3-2}
\frac{\partial f}{\partial {y}}+\theta_1\frac{\partial g}{\partial {y}}&=- {\theta_2} \\
\label{kkt3-3}
\frac{\partial f}{\partial {z}}+\theta_1\frac{\partial g}{\partial {z}}&=0
\end{align}
Since $y=\frac{x+1}{2x+1}$, Lemma \ref{lem:9}(i) asserts $\frac{\partial f}{\partial {x}}=\frac{\partial f}{\partial {y}}$ and $\frac{\partial g}{\partial {x}}=\frac{\partial g}{\partial {y}}$, which along with (\ref{kkt3-1}) and (\ref{kkt3-2}) enforce $\theta_2=0$. In turn from (\ref{kkt3-1}) and (\ref{kkt3-3}) we derive
$\frac{\partial f/\partial {x}}{\partial g/\partial {x}}=-\theta_1=  \frac{\partial f/\partial {z}}{\partial g/\partial {z}}$. By ${y}=\frac{{x}+1}{2{x}+1}$, Lemma \ref{lem:9}(iv) enforces ${x}=\chi$. Hence from $g({x},{y},{z})=0$ we obtain
\[{z} =
\frac12\left(1 +\beta-\frac{1+2{\chi} }h+\frac1{{\chi}}-
\left(\frac2h+\frac1{{\chi}({\chi} +1)} \right)\cdot\frac{\chi+1}{2\chi+1} \right)=:z_3
\]
It follows that $\Omega_4=f(\chi,
\frac{\chi+1}{2\chi+1},z_3)=\frac{4\sqrt{2h^2-h+1}
+1}{2h}+\frac{\beta-5}2$, which increases in $h$ for  $h\geq7$. So
$\Omega_4$ is lower bounded by $\frac{4\sqrt{92}+1}{14}
+\frac{\beta-5}2$, which is greater than $\frac{1+\beta}3$ if
$\beta>\frac{116-12\sqrt{92}}7>0.13$. Since $\Omega_2\ge\Omega_4$ in
either case, we have shown that
 \be\label{013}
 \Omega_2>(1+\beta)/3\,\text{ if }\,\beta\ge0.13.
 \ee

Next we again focus on the optimal solution $({x}^*,{y}^*,{z}^*)$ of (\ref{o2}).
 \begin{claim}\label{case3}
If ${x}^*>1$ and $0<y^*<1$, then $\Omega_2\ge(1+\beta)/3$.
\end{claim}
  \begin{proof}  If ${z}^*>\beta/h$, by  ${x}^*>1$ and the KKT conditions   (\ref{kt})--(\ref{kt2}),   we obtain $\mu_1=\mu_4=0$ and $\frac{\partial f/\partial {x}}{\partial g/\partial {x}}=-\lambda
= \frac{\partial f/\partial {z}}{\partial g/\partial {z}}$ at $({x}^*,{y}^*,{z}^*)$, which is equivalent to ${x}^*=\chi$ by ${y}^*=\frac{{x}^*+1}{2{x}^*+1}$ and Lemma \ref{lem:9}(iv). In turn we have ${z}^*=z_3$ by using $g(\chi, \frac{\chi+1}{2\chi+1},z^*)=0$. Now  $z_3>\beta/h$ reads
$1+\beta-\frac{1+2\beta}{h}>\frac{\sqrt{2h^2-h+1}}{h}-\frac{2}{\sqrt{2h^2-h+1}}+\frac{1}{h\sqrt{2h^2-h+1}}$. The right-hand side
of this inequality is larger than $\frac{\sqrt{2h^2-h}}{h}-\frac{2}{\sqrt{2h^2-h}}$. It follows from Claim \ref{rm} that   $\beta>\sqrt{2}-1>0.13$, and further from (\ref{013}) that   $\Omega_2>(1+\beta)/3$.

If ${z}^*= {\beta}/h$,
by (\ref{013}), it suffices to consider   $\beta\in[0,0.13)$. From $g(x^*,{y}^*,z^*)=g(x^*,\frac{{x}^*+1}{2{x}^*+1},\frac{\beta}h)=0$ we get
{\[h=\frac{(2x^*+1)^2+(2\beta+1)(2x^*+1)+1}{(\beta+1)(2x^*+1)+2}.\]}
Under this equation for $h\ge7$ and   $\beta\in[0,0.13)$, Lemma \ref{lemb} asserts $$\Omega_2=f({x}^*,{y}^*,{z}^*)=\frac{h}{{x}^*}+\frac{{x}^*}{h}-\left(\frac{h}{{x}^*({x}^*+1)}-\frac1h\right)\frac{{x}^*+1}{2{x}^*+1}
-\frac{\beta}h-2+\beta>0.36+\beta,$$ which is obviously greater than $\frac{1+\beta}3$.
\end{proof}

To sum up, we have shown the following result.
\begin{lemma}
\label{lem:firstratio3}
If $\cup_{i \in [h]}N_i= R$ and $h>2$, then $\ell^N(R)/M(\pi^*)\leq 3$.
\end{lemma}

\section{The Cases $h=1$ and $h=2$}\label{sec:h=12}
\begin{lemma}\label{lem:h1}
Let $\mathcal I$ be an SRR instance with $h=1$. Then $M(\pi^N)/M(\pi^*)\le2$.
\end{lemma}
\begin{proof}
Only the case $k=2$ is relevant.
By assumption, we have $N_1 = R\backslash Q_1$ and we have $N_2=Q_2$.
That is, we have $\pi^*=\{Q_1, Q_2\}$ and we have $\pi^N=\{N_1, Q_2\}$.
In particular, we have
$\ell^N(Q_1)+||Q_1||_a = \ell^*(Q_1)$. Hence,
$\ell^N(Q_1)+||Q_1||_a \le M(\pi^*)$.
Since $\pi^N$ is a Nash equilibrium, it holds that
$\ell^N(N_1) \leq \ell^N(Q_1)+||Q_1||_a \leq M(\pi^*)$.
We thus get
\begin{equation}
\label{eq:h=1}
M(\pi^N)
\leq    \ell^N(R)
=           \ell^N(N_1)+\ell^N(Q_1)
\leq    2M(\pi^*)\,,
\end{equation}
as desired.
\end{proof}

\label{app:h=2}
\begin{lemma}
\label{lem:uncovered2}
Let $\mathcal{I}$ be an SRR instance with $h=2$ and $\cup_{i=1,2} N_i\ne R$.
Then $M(\pi^N)/M(\pi^*)\le2$.
\end{lemma}

\begin{proof}
We first consider the case that $\mathcal{I}$ is nonsingular.
Then $k=h=2$ by definition.
Assume without loss of generality that $\ell^N(N_1) \leq \ell^N(N_2)$.
Since $\pi^N$ is a Nash equilibrium, it holds that
$\ell^N(N_2)\leq \ell^N(Q_2)+||Q_2||_a$, which, by the fact that $k=h=2$, is at most $2\ell^*(Q_2)$.
Therefore,
$M(\pi^N) = \ell^N(N_2) \leq 2\ell^*(Q_2) \leq 2 M(\pi^*)\,.$

Let us now consider the case that $\mathcal{I}$ is singular.
That is, we have $k=3$ and $M(\pi^N)=\ell^N(N_3)$.
First note that we can rewrite
$$
Q_1
= R\backslash N_1
= (N_2 \cup Q_2)\backslash N_1
= (N_2 \backslash N_1) \cup (Q_2\backslash N_1)
= (N_2\backslash (N_1 \cap N_2)) \cup (Q_1 \cap Q_2)
$$
and, similarly, we have $Q_2=(N_1\backslash (N_1 \cap N_2)) \cup (Q_1 \cap Q_2)$.
Using this and the fact that $\pi^N$ is a Nash routing, we obtain the following two inequalities.
\begin{align*}
\ell^N(N_1)&\leq\ell^N(Q_1)+||Q_1||_a=\ell^N(N_2)-\ell^N(N_1\cap N_2)+\ell^N(Q_1\cap Q_2)+||Q_1||_a\\
\ell^N(N_2)&\leq\ell^N(Q_2)+||Q_2||_a=\ell^N(N_1)-\ell^N(N_1\cap N_2)+\ell^N(Q_1\cap Q_2)+||Q_2||_a
\end{align*}
From this we get
$2\ell^N(N_1\cap N_2)\leq2\ell^N(Q_1\cap Q_2)+||Q_1||_a+||Q_2||_a.$
Furthermore, when comparing $\ell^*(Q_i)$ and $\ell^N(Q_i)$ for $i=1,2$, we can ignore
player 3, because it contributes the same to both values. Hence
\begin{align*}
\ell^*(Q_1)
& = \;\ell^*(Q_1 \cap N_2) + \ell^*(Q_1 \cap Q_2)\\
&=\; \ell^N(Q_1 \cap N_2) + \ell^N(Q_1 \cap Q_2) + 2 ||Q_1\cap Q_2||_a\\
& =\; \ell^N(Q_1) + 2 ||Q_1\cap Q_2||_a,
\end{align*}
and, similarly, $\ell^*(Q_2)=\ell^N(Q_2) + 2 ||Q_1 \cap Q_2||_a$ holds.
From this we conclude
\begin{align}
 M(\pi^*)
& \geq \ell^*(Q_1)
=\ell^N(Q_1)+2||Q_1\cap Q_2||_a \nonumber \\
& =\ell^N(Q_1\cap Q_2)+\ell^N(Q_1\backslash Q_2)+2||Q_1\cap Q_2||_a \label{optn1}\\
M(\pi^*)
& \geq \ell^*(Q_2)
=\ell^N(Q_2)+2||Q_1\cap Q_2||_a \nonumber \\
&=\ell^N(Q_1\cap Q_2)+\ell^N(Q_2\backslash Q_1)+2||Q_1\cap Q_2||_a\label{optn2}
\end{align}
and
$
M(\pi^*) \ge \ell^*(N_3) \geq \ell^N(N_3)-2||N_3\cap N_1\cap N_2||_a. 
$
Notice that
$\ell^N(Q_1\backslash Q_2) \geq ||Q_1\backslash Q_2||_a$,
$\ell^N(Q_2\backslash Q_1) \geq ||Q_2\backslash Q_1||_a$
and
$2||N_3\cap N_1\cap N_2||_a\leq\ell^N(N_1\cap N_2)$.
We may thus conclude $M(\pi^N)
 = \ell^N(N_3)$ is upper bounded by
\begin{align*}
 & \;M(\pi^*) + 2 ||N_3\cap N_1\cap N_2||_a \\
 \leq&\; M(\pi^*) + \ell^N(N_1 \cap N_2)\\
 \leq & \;M(\pi^*) + \ell^N(Q_1 \cap Q_2) + \frac{1}{2} ||Q_1||_a + \frac{1}{2} ||Q_2||_a \\
 \leq & \;2 M(\pi^*) - \frac{1}{2} \ell^N(Q_1\backslash Q_2) - 2 ||Q_1 \cap Q_2||_a
                                    -   \frac{1}{2} \ell^N(Q_2\backslash Q_1)
                                    + \frac{||Q_1||_a + ||Q_2||_a}{2}\\
 \leq &\;2 M(\pi^*) - \frac{1}{2} ||Q_1 \backslash Q_2||_a - 2||Q_1 \cap Q_2||_a
                                    - \frac{1}{2} ||Q_2 \backslash Q_1||_a +  \frac{||Q_1||_a + ||Q_2||_a}{2}\\
  \leq & \;2 M(\pi^*). \hspace{100mm}
\end{align*}

\end{proof}

\begin{lemma}
\label{lem:firstratio2}
If $\cup_{i \in [h]}N_i= R$ and $h\le 2$, then $\ell^N(R)/M(\pi^*)\leq 3$.
\end{lemma}

\begin{proof}
When $h=1$, inequalities (\ref{eq:h=1}) imply the conclusion. By Lemma \ref{lem:singular}, it remains to consider  $k=h=2$ and $N_1 \cup N_2 = R$.
Suppose without loss of generality that $\ell^N(Q_{1})\le \ell^N(Q_{2})$.
Note that $Q_2 \subseteq N_1$ and thus $N_1= Q_2 \cup (N_1 \cap N_2)$. This yields
$\ell^N( Q_2)+\ell^N(N_1\cap N_2) = \ell^N(N_1) \le \ell^N(Q_1)+||Q_1||_a$, where the latter inequality stems from the fact that player $1$ does not want to deviate in $\pi^N$.
Together with the assumption $\ell^N(Q_1)\le \ell^N(Q_2)$ we thus have
  $\ell^N(N_1\cap N_2) \leq \ell^N(Q_1)-\ell^N(Q_2) + ||Q_1||_a \leq ||Q_1||_a \leq M(\pi^*)$. It follows from $N_i=R\backslash Q_i$, $i=1,2$ that
\begin{align*}\ell^N(R)
 &=\ell^N(Q_2)+\ell^N(Q_1) + \ell^N(N_1\cap N_2)\le\ell^N(Q_2)+\ell^N(Q_1)+M(\pi^*).\end{align*}
Since $\ell^N(Q_i)=\ell^*(Q_i)$ for $i=1,2$, we obtain
$\ell^N(R) \leq\ell^*(Q_1)+\ell^*(Q_2)+M(\pi^*)   \leq 3 M(\pi^*)$ as desired.
\end{proof}

\section{Covering Equilibria with $h=5$}\label{adx5}

For the special case of $h=5$, we upper bound $\ell^N(R)/M(\pi^*)$
directly by using structural properties of the Nash equilibrium.

\begin{lemma}
\label{lem:hequal5}
If $\cup_{i \in [h]}N_i= R$ and $h=5$, then $\frac{\ell^N(R)}{M(\pi^*)}\le3$.
\end{lemma}

\begin{proof}
Again, by Lemma \ref{lem:singular}, we only need to consider the case where $\pi^N$ is nonsingular.
If $\pi^N(e)\ge2$ for all $e\in E$, then $A_1=B_1=0$ in   (\ref{eq:eibound}) and (\ref{eq:lnrmf}). Collecting terms in (\ref{eq:eibound}) gives $A_3+11A_4+25A_5+B_3+3B_4+5B_5\le5A_2+B_2$, which is equivalent to $5(\sum_{i=2}^5iA_i+B_i)\le25A_2+12A_3-13A_4-50A_5 +8B_2+2B_3-4B_4-10B_5$. It follows from (\ref{eq:lnrmf}) that
\begin{align*}
\frac{\ell^N(R)}{M(\pi^*)}&\le\frac{5(\sum_{i=2}^5iA_i+B_i)}{\sum_{i=2}^5((5-i)^2A_i+(5-i)B_i)}\\
& \le \frac{25A_2+12A_3 +8B_2+2B_3 }{9A_2+4A_3+A_4+4B_1+3B_2+2B_3+B_4}\leq 3.
\end{align*}

Therefore, we may assume without loss of generality that there exists a link $e_1\in N_1$ with $\pi^N(e_1)=1$. Note that this implies $e_1\notin \cup_{i=2}^{5}N_i$ and, thus, $\cup_{i=2}^{5}N_i \ne R$.
Starting from link $e_1$, let $v$ be the clockwise first node
where the Nash path $N_i$ of another player $i\in\{2,3,4,5\}$  starts.  For the analysis,  let us temporarily split the ring at node $v$, and put the nodes
on a line from left to right, starting and ending with $v$.
Then for each player in $\{2,3,4,5\}$, its Nash path is one line segment
by $\cup_{i=2}^{5}N_i\ne R$ and definition of $v$. See Figure \ref{fg:open} for an illustration.

Let $F \subseteq \{2,3,4,5\}$ consist of  two players with the leftmost left endpoints,
and $L \subseteq \{2,3,4,5\}$ consist of two agents with the rightmost right endpoints.
(Going from left to right, $F$ are two of the \emph{first} players that start, and $L$ are two of the
\emph{last} players that finish their Nash paths.)

\begin{figure}[!t]
\centerline{\includegraphics[scale=1]{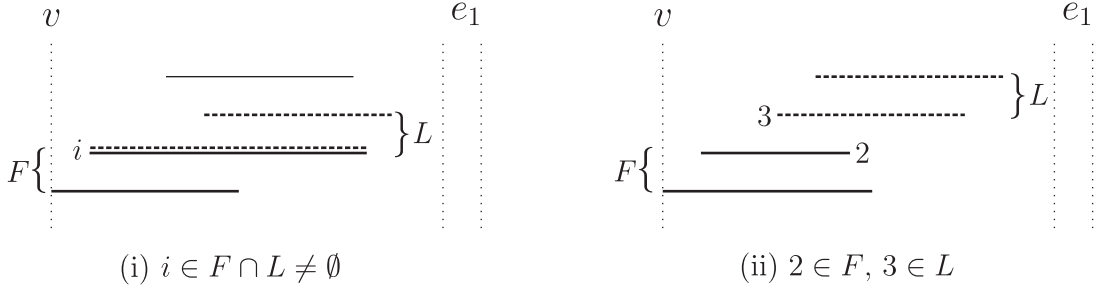}}
\caption{ \label{fg:open} Splitting $R$ at node $v$, where the Nash path $N_1$ containing $e_1$ is not depicted. }
\end{figure}

If there exists a player $i$ that is in both $F$ and $L$---formally
if $i\in F\cap L\not=\emptyset$ (see Figure \ref{fg:open}(i) for an
illustration), then the definitions of $F$ and $L$ guarantee that
both to the left and to the right of the path $N_i$ of $i$ in
$\pi^N$,  any link can only be used (in $\pi^N$) by   at most one
player $j\in\{2,3,4,5\}$ and possibly by the first player. It
follows that $\pi^N(e)\le2$ and hence $3\le\pi^*(e)$ for all $e\in
Q_i$. In particular we have $\pi^N(e)+1 \leq \pi^*(e)$ for all $e\in
Q_i$. Since $\pi^N$ is a Nash equilibrium, we conclude that
$\ell^N(N_i)\le\ell^N(Q_i)+||Q_i||_a\le\ell^*(Q_i)\le M(\pi^*)$,
giving $\ell^N(R)=\ell^N(N_{i})+\ell^N(Q_{i})\leq 2M(\pi^*)$ as
desired.

Therefore, let us consider the case $F\cap L=\emptyset$. Without loss of generality, let player 2 be a player in $F$ with the rightmost right endpoint,
and let player 3 be a player in $L$ with the leftmost left endpoint (an illustration is given by Figure \ref{fg:open}(ii)).
Then in $\pi^N$, any link to the right of $N_2$ can only be used by players in $L\cup\{1\}$, i.e., it can be used by at most three players, and any link to the left of $N_2$ can
only be used by players in $(F\backslash\{2\})\cup\{1\}$, i.e., it can be used by at most two players.
Thus $\pi^N(e)\le3$ for every $e\in  Q_2 $. Analogously we have $\pi^N(e)\le3$ for every $e\in  Q_3 $.
Moreover, we see that $\{e\in Q_2|\pi^N(e)=3\}\subseteq  N_3 \cap   N_1 $ and
$\{e\in  Q_3 |\pi^N(e)=3\}\subseteq  N_2 \cap   N_1 $. On the other hand, from the selections of player 2 from $F$, and player 3 from $L$, it is easy to see that $\{e\in  Q_1 |\pi^N(e)\ge3\}\subseteq N_2\cap N_3$.
This implies the following useful inequality, valid for all assignments $\{r,s,t\}=\{1,2,3\}$:
\begin{align}
\label{useful}
\ell^N(Q_r\backslash (N_s\cap N_t))+||Q_r\backslash (N_s\cap N_t)||_a
\leq \ell^*(Q_r\backslash (N_s\cap N_t))\,.
\end{align}

Let $S=\{(1,2,3),(2,3,1),(3,1,2)\}$.
Adding the Nash inequalities for the paths $N_i$, $i=1,2,3$, and their alternatives gives
\begin{align*}
  \sum_{i=1}^3\ell^N(N_i) \le& \sum_{i=1}^3(\ell^N(Q_i)+||Q_i||_a) \\
    =& \sum_{(r,s,t)\in S}\left[\ell^N(Q_r\cap N_s\cap N_t)+\ell^N(Q_r\backslash( N_s\cap N_t))\right]\\
    &+\sum_{(r,s,t)\in S}\left[||Q_r\cap N_s\cap N_t||_a+||Q_r\backslash( N_s\cap N_t)||_a\right]
\end{align*}
It follows from (\ref{useful}) and $||Q_1\cap N_2\cap N_3||_a\le\ell^*(Q_1\cap N_2\cap N_3)$ that
\begin{align*}
&\qquad  \sum_{i=1}^3\ell^N(N_i) \\
&\le  \sum_{(r,s,t)\in S}\left(\ell^N(Q_r\cap N_s\cap N_t)+\ell^*(Q_r\cap N_s\cap N_t)+\ell^*(Q_r\backslash( N_s\cap N_t))\right)\\
   &=  \sum_{(r,s,t)\in S} \ell^N(Q_r\cap N_s\cap N_t)+\sum_{i=1}^3\ell^*(Q_i)\\
    &=  \sum_{(r,s,t)\in S} \left(\ell^N( N_s\cap N_t)-\ell^N( N_r\cap N_s\cap N_t)\right)+\sum_{i=1}^3\ell^*(Q_i)\\
    &= \sum_{1\le i<j\le3}\ell^N(N_i\cap N_j)  -3\ell^N(N_1\cap N_2\cap N_3)+\sum_{i=1}^3\ell^*(Q_i)\,,
\end{align*}
thus implying
\[ \sum_{i=1}^3\ell^N(N_i)-\sum_{1\le i<j\le3}\ell^N(N_i\cap N_j)  +\ell^N(N_1\cap N_2\cap N_3)\le \sum_{i=1}^3\ell^*(Q_i)\]
Notice that the left-hand side of the above inequality equals $\ell^N(R)$ and its the right-hand side is at most $3M(\pi^*)$.  The result follows.
\end{proof}

\section{Concluding Remarks}
We have shown that the PoA of network congestion game is two, when
the network is a ring and the link latencies are linear.
It is left open whether the PoA is exactly $2^d$
for polynomial latency functions of degree $d$. Another challenging open question
is what happens in more complicated network topologies. It is interesting to see if our proof technique can be extended to the more general class of games where each player can choose between a set of resources and its complement.

\subsubsection*{Acknowledgments.}
This work would not have been started without Xujin Chen and Benjamin Doerr
having been invited to the Sino-German Frontiers of Science Symposium
(Qingdao, 2010) organized by the Humboldt Foundation. A follow-up visit
of Xujin Chen at the MPI in 2010 was also funded by the Humboldt Foundation
via a CONNECT grant. We thank the Humboldt Foundation for providing
both means of support.

This work was also supported in part by NNSF of China under Grant No. 11222109, 11021161 and 10928102,
by 973 Project of China under Grant No. 2011CB80800, by CAS under Grant No. kjcx-yw-s7.

Rob van Stee would like to thank \'Eva Tardos for interesting discussions.

Carola Winzen is a recipient of the Google Europe Fellowship in Randomized Algorithms. This research is supported in part by this Google Fellowship.

\bibliographystyle{amsalpha}
\bibliography{references}
\end{document}